\documentclass[10 pt,final,journal,letterpaper,oneside,twocolumn]{IEEEtran}
\usepackage[utf8]{inputenc}
\usepackage[T1]{fontenc}
\usepackage{physics,amsmath}
\usepackage{amsthm}
\usepackage{amsfonts}
\usepackage{mathtools}
\usepackage{amsfonts}
\usepackage{amssymb}
\usepackage{makeidx}
\usepackage{csquotes}
\usepackage{graphicx}
\usepackage{cite}
\usepackage{xcolor}
\usepackage{subcaption}
\usepackage{graphicx}
\usepackage{multicol}
\usepackage{float}
\usepackage{multicol}

\usepackage{multirow}
\usepackage{array}
\usepackage{pifont}
\usepackage{multirow}
\usepackage{hyperref}
\usepackage[capitalise,nameinlink ,noabbrev]{cleveref}

\newtheorem{lemma}{Lemma}
\usepackage[algo2e,ruled, vlined, linesnumbered]{algorithm2e}
\newtheorem{remark}{Remark}
\begin{document}
\title{Joint Computation and Communication Resource Optimization for Beyond Diagonal UAV-IRS Empowered MEC Networks}
\author{\IEEEauthorblockN{Asad Mahmood, Thang X. Vu, \IEEEmembership{Senior Member,~IEEE}, Wali Ullah Khan, \IEEEmembership{Member,~IEEE}, Symeon Chatzinotas, \IEEEmembership{Fellow,~IEEE}, and 
Bj\"orn Ottersten, \IEEEmembership{Fellow,~IEEE}, }
\thanks{{The part of this work has been presented at the  2022 IEEE Globecom Workshops (GC Wkshps) \cite{10008627}}. This work is supported by the Luxembourg National Research Fund via project 5G-Sky, ref. C19/IS/13713801/5G-Sky, and project RUTINE, ref. C22/IS/17220888/RUTINE. \newline   
\indent The authors are with the Interdisciplinary Centre for Security,
Reliability and Trust (SnT), University of Luxembourg, 4365 Luxembourg
City, Luxembourg. Email: \{asad.mahmood, thang.vu,waliullah.khan symeon.chatzinotas, bjorn.ottersten\}@uni.lu.}}
\markboth{Submitted to Journal}%
{Shell \MakeLowercase{\textit{et al.}}: Bare Demo of IEEEtran.cls for IEEE Journals} 
\maketitle
\begin{abstract}
Recent advances in sixth-generation (6G) communication systems mark a critical evolution towards widespread connectivity and ultra-reliable, low-latency communication for numerous devices engaged in real-time data collection. However, this progression faces significant challenges, primarily due to limited battery life and computational capacity of devices, and signal obstruction from urban infrastructure like high-rise buildings. Intelligent Reconfigurable Surfaces (IRS) integrated within Mobile Edge Cloud (MEC) infrastructures emerge as a strategic response, offering enhanced computing on demand to navigate device constraints and alternative communication routes to alleviate coverage issues. Nonetheless, connectivity challenges persist for users in remote or less populated areas, restricting their access to consistent and reliable communication networks.
This paper introduces Beyond Diagonal IRS (BD-IRS, or IRS 2.0), a new IRS family member deployed on unmanned aerial vehicles (UAVs) within Mobile Edge Computing (MEC) networks, termed BD-IRS-UAV. This system offers on-demand links for remote users to offload tasks to the MEC server, addressing resource and battery constraints. We develop a joint optimization strategy to minimize the system's worst-case latency and UAV hovering time by optimizing BD-IRS-UAV deployment and the allocation of computational and communication resources. The optimization challenge is addressed by decomposing it into two sub-problems: 1) BD-IRS-UAV Placement and Computational Resource Optimization, and 2) Communication Resource Optimization, solved through an iterative process. The proposed design markedly improves system performance, exhibiting a superior rate increase of $17.75\%$ over traditional single-connected diagonal IRS and a $25.43\%$ enhancement compared to  IRS placements on buildings. Additionally, by considering worst-case latency as a performance metric, our proposed scheme exhibits an enhancement of approximately $13.44\%$ relative to the binary offloading scheme.

\end{abstract}
\section{Introduction}
\IEEEPARstart{I}{n} the rapidly advancing landscape of wireless communications, the emergence of the next-generation wireless networks holds great promise and equally formidable challenges. As we stand on the edge of the 6G era, there is a collective expectation for wireless networks to deliver high-speed data rates, pervasive connectivity, robust reliability, and ultra-low latency \cite{santos2023automated,9861699}. These higher standards result from the digital age's more varied and data-hungry applications, ranging from the Internet of Things (IoT) ecosystem to augmented reality and virtual reality experiences \cite{mahmood2021optimal}. However, unlocking the full potential of these future networks is challenging. The anticipated proliferation of connected devices, data-hungry applications, and mission-critical use cases introduces various challenges, e.g., computation burden and stringent latency. Traditional wireless networks may need help to cope with the surge in demand for {\color{black}high data rates}, and the existing infrastructure may need to improve in delivering the ultra-low latency required for real-time applications \cite{wang2023online}. 
Mobile edge computing (MEC) has emerged as a promising technology with the potential to tackle these challenges \cite{santos2023automated}. By bringing computational resources closer to the users, MEC minimizes latency, optimizes traffic routing, and enhances computational energy efficiency, improving overall network efficiency \cite{li2021intelligent}. However, the deployment of MEC networks and their effectiveness can be further amplified through innovative approaches, and herein lies the synergy with unmanned aerial vehicles equipped with intelligent reflecting surfaces (IRS-UAV) \cite{zhai2022energy,9768113}. The integration of IRS-UAV into the MEC framework introduces unprecedented flexibility and adaptability. IRS-UAV can dynamically adjust their positions and manipulate signals, optimizing MEC services coverage and capacity. This combination of IRS-UAV and MEC networks {\color{black}presents a promising architecture for future} 6G wireless communications \cite{he2022joint}.
\par
IRS, a key element in 6G networks, can be reconfigured dynamically to improve signal strength, energy efficiency, and quality in wireless networks \cite{chen2022multi,8811733}. IRS can generally be categorized based on operating modes, which include reflective, transmissive, and hybrid modes, and architectural designs, which comprise single-connected, fully-connected, and group-connected \cite{li2023beyond,10159457}. The following sections provide an in-depth exploration of the various types and categories of IRS, offering more comprehensive insights. As such, recognizing the significance of joint UAV-IRS deployment and resource allocation becomes paramount, particularly in scenarios where low latency is imperative.
\subsection{Related Works} 
Researchers have extensively studied IRS-enabled MEC networks. The research works in \cite{xu2022energy,el2021latency,li2021energy,mao2022reconfigurable,li2021intelligent} investigated the energy efficiency of the system. The problems in those studies are non-convex in nature due to interference terms and coupling decision variables. Such problems are very challenging to obtain their optimal solutions. Thus, they used alternating optimization method to achieve suboptimal solutions. Moreover, these works have optimized various variables to improve the system performance. In particular, they have optimized transmit power, time allocation, offloading partitions, decoding order, phase shift design, transmission rate, and size of transmission data. 
Moreover, Chen {\em et al.} \cite{chen2022multi} have considered downlink and uplink transmissions in IRS enabled wireless power transfer MEC networks. The authors optimized energy beamforming, user detection, transmit power, and phase shift design to maximize the sum computational rate of the system. Furthermore, the authors in \cite{liu2023star} have also considered simultaneous transmitting and reflecting IRS in MEC networks. They optimized the phase shift design, energy partition strategies for local computing and offloading, and the received beamforming. In addition, Zhu {\em et al.} \cite{zhu2022dynamic} have considered an IRS enabled MEC network in the vehicular scenario and optimized the computational resources and IRS phase shift to improve the system's task offloading, computing, and finish rates. 
\par
In addition, some research studies have employed artificial intelligence-based algorithms in IRS enabled MEC networks \cite{wang2023online,hu2021reconfigurable,huang2021reconfigurable,wang2022resource}. In particular, the work in \cite{wang2023online} has adopted an online algorithm named Lyapunov optimization in combination with semi-definite relaxation for computational energy efficiency. In \cite{hu2021reconfigurable}, the authors have maximized the total completed tasks using deep learning algorithm. Accordingly, \cite{huang2021reconfigurable} has exploited machine learning techniques to minimize the maximize learning error in MEC networks while \cite{wang2022resource} has used deep reinforcement learning for optimizing energy consumption. Of late, Huang {\em et al.} \cite{huang2022integrated} have proposed integrated sensing and communication in IRS enabled heterogeneous MEC networks. Following that, the energy consumption minimization problem of the UAV integrated with the IRS MEC network is discussed in \cite{9651523}
\par   
Besides the above literature, few studies have considered UAV communication in IRS enabled MEC networks. For instance, Zhai {\em et al.} \cite{zhai2022energy} have considered IRS-UAV in MEC networks. The authors maximized the achievable energy efficiency of the system by jointly optimizing MEC resources, passive IRS beamforming and UAV trajectory. In another work \cite{savkin2022joint}, the authors have considered multiple UAVs in IRS enabled MEC networks. They enhanced the computational energy efficiency of the system by jointly optimizing the UAV's path planning and transmission. The research work in \cite{shnaiwer2022minimizing} has jointly optimized the IRS phase shift and UAV path selection to maximize the computation energy efficiency of IRS-UAV MEC networks. Moreover, the authors in \cite{shnaiwer2022multihop} have provided a graph-based heuristic algorithm for multi-hop task routing in IRS enabled MEC networks involving UAV communications. Their objective was to improve the average energy efficiency of the system. The authors of \cite{he2022joint} have investigated a latency problem in the IRS-UAV MEC network by optimizing the time allocation, IRS passive and user active beamforming, UAV hovering, and computing task scheduling. Furthermore, Asim {\em et al.} \cite{asim2022energy} have jointly optimized the UAV trajectory and IRS phase shift to maximize the energy efficiency in MEC networks. 
\par    
Recently, researchers have proposed beyond diagonal IRS (BD-IRS, also known IRS 2.0) various wireless networks. For instance, in \cite{li2022reconfigurable}, the authors have maximized the achievable data rate of a multi-user systems by jointly optimizing the transmit beamforming and beyond diagonal phase shift design in {\color{black}BD-IRS} networks. Another work in \cite{nerini2023closed} has derived a closed-form solution to optimize the received signal power of {\color{black}BD-IRS} enabled single user single-input single-output, single-user multiple-input multiple-output, and multi-user multiple-input single-output systems. Moreover, Li {\em et al.} \cite{li2022beyond} have investigated a sum rate maximization problem in {\color{black}BD-IRS} enabled multi-user wireless systems by jointly optimizing the transmit precoding and beyond diagonal phase shift design. In \cite{nerini2023pareto}, the researchers have derived the expression of the Pareto frontier to investigate the complexity trade-off in {\color{black}BD-IRS} systems. Besides that, the authors of \cite{li2023beyond} have proposed a multi-sector {\color{black}BD-IRS} communication model and designed an efficient beamforming algorithm to maximize the sum rate of the system. Further, the same authors in \cite{10159457} have provided an efficient algorithm for dynamic grouping to maximize the {\color{black}BD-IRS} enabled multi-user systems.

\subsection{Motivation and Contributions}
As highlighted in the preceding discussion, MEC, assisted by the UAV-enabled IRS, emerges as an effective solution to provide effective communication through an assisted communication link to users located in remote areas in dense urban environments. For this, extensive research has been done on diagonal IRS in MEC systems, including work by \cite{zhu2022dynamic,xu2022energy,li2021intelligent,el2021latency,chen2022multi,li2021energy,mao2022reconfigurable,liu2023star,wang2023online,hu2021reconfigurable,huang2021reconfigurable,wang2022resource,huang2022integrated,9651523}. Accordingly, the authors in \cite{zhai2022energy,savkin2022joint,shnaiwer2022minimizing,shnaiwer2022multihop,he2022joint,asim2022energy} have proposed UAV communication in MEC networks. More closely related work to ours was study in \cite{9651523} addresses network energy optimization through task segmentation, transmission power, trajectory, and IRS phase shift adjustments with the assumption of constant device energy use with fixed local computational capabilities—an approach divergent from real-world scenarios where device energy consumption closely aligns with available computational resources. On the other hand, the research described in \cite{li2022reconfigurable,nerini2023closed,li2022beyond,nerini2023pareto,li2023beyond,10159457} delves into integrating BD-IRS into conventional network configurations, overlooking its compatibility with MEC and UAV communications, as well as the essential consideration of optimal IRS placement. To the best of our knowledge, this is the first work that exploits the BD-IRS in combination with UAVs in MEC systems. Our goal is to minimize the worst-case latency of the system via a joint design of the optimal location of BD-IRS-UAV, local and edge computation resources, task segmentation, allocation of user transmit power, received beamforming vector, and phase shift design of the BD-IRS. Due to multiple optimization variables, the joint optimization problem is non-convex, non-linear, and NP-hard. Therefore, we first divide the original problem into two subproblems, and then each subproblem is transformed into a tractable representation, which can be efficiently solved using standard convex optimization methods.

\par Our main contributions in this paper can also be summarized as follows.
\begin{enumerate}
\item We consider a BD-IRS-UAV enabled MEC network where multiple users are intended to offload their tasks to the MEC server due to limited resources and finite battery life. The MEC receives signals from direct and BD-IRS-UAV links. Our aim is to minimize the worst-case computation latency and UAV hovering time via a joint design of the system's parameters.  
\item To tackle the prohibitive computational complexity of the original joint optimization problem, we first decouple it into two subproblems: the first subproblem designs the BD-IRS-UAV location, computation resource allocation, and task segmentation, while the second problem optimizes the transmit power, beamforming vectors, and phase shift.
\item In particular, to solve the first sub-problem, we propose a heuristic K-Mean clustering algorithm to determine the location of the BD-IRS-UAV based on its placement constraints. A closed-form equation for the task segmentation variable is calculated, which is then solved iteratively with a computational resource allocation problem to minimize the maximum latency. To tackle the non-convexity of the rate function in the second sub-problem, a Lagrange dual transformation-based approach is applied, for which an iterative algorithm is proposed. We show that the convergence to at least a local optimum of the proposed iterative algorithm is theoretically guaranteed.
\item Comprehensive Monte Carlo simulation results are provided to validate the efficacy of the proposed framework. These results reveal that both fully-connected and group-connected architectures surpass the traditional single-connected architecture by $23.45\%$ and $17.75\%$, respectively, when evaluating performance based on rate. Furthermore, the optimal deployment of IRS achieves an approximate improvement of $25.76\%$ compared to conventional IRS placement on buildings. Additionally, by considering worst-case latency as a performance metric, our proposed scheme exhibits an enhancement of approximately $13.44\%$ relative to the binary offloading scheme.
\end{enumerate}
\subsection{Organization}
The following sections summarize the rest of the paper:
Section \ref{System Model} summarizes the proposed system model and problem statement. Section \ref{Proposed Solution} presented an algorithm for efficient BD-IRS-UAV deployment and resource allocation. Likewise, Section \ref{RD} contains the simulation results. The paper is finally concluded in Section \ref{CC}
\section{System Model}
\label{System Model}
We consider an up-link BD-IRS-UAV-enabled MEC network in a smart city, where multiple low-powered user devices (UEs) are randomly distributed throughout the area to sense and collect real-time application-generated data. Let us denote the number of UEs as $\mathcal{N}=\{n|1,2,3,\cdots N\}$, where $n$ indexes $n$-th UE. The access point (AP) in the MEC system is equipped with $M$ antennas. Considering the obstacles on the surface (i.e., buildings, trees, vehicles, etc.), the received signal may be affected, resulting in overall system performance degradation. To overcome this issue and further enhance the performance, a BD-IRS-UAV consisting of $K$ cells, denoted as $\mathcal{K}=\{k|1,2,3,\cdots K\}$, is positioned with a height of $H$ to reflect the signal toward MEC. Thus, communications between UEs and the MEC occur through direct and BD-IRS-UAV-assisted links.
\begin{figure}
\centering
\includegraphics[width=0.46\textwidth]{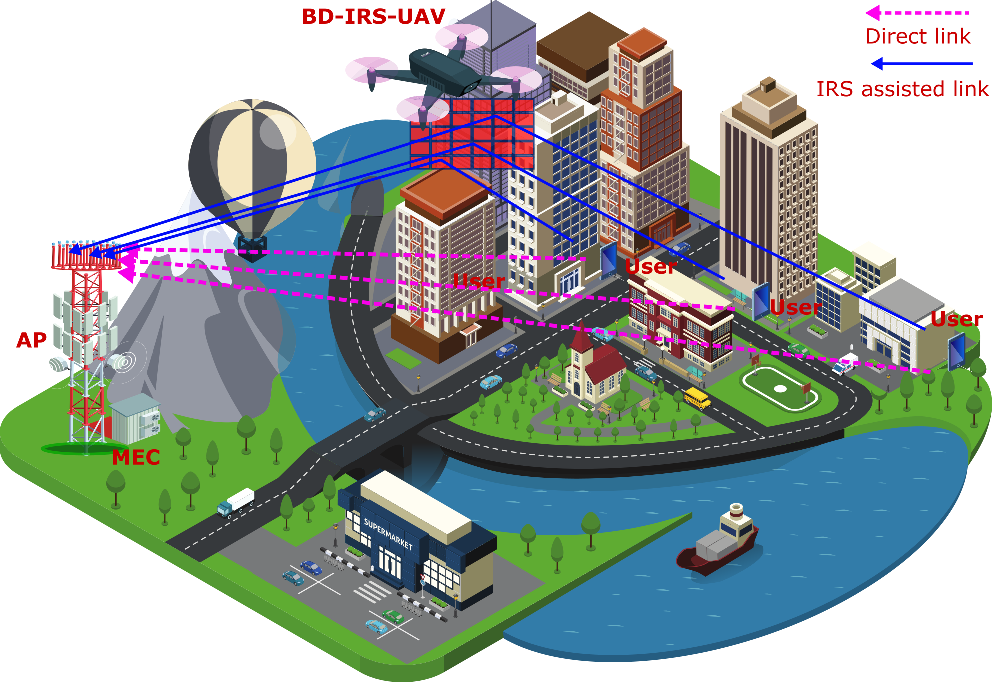}
\caption{System model}
\vspace{-5mm}
\label{SM}
\end{figure}
\par
To ensure clarity, we represent the positions of the BD-IRS-UAV, UEs, and AP within the 3D Cartesian coordinate system: $\textbf{z}=[x_u, y_u, H]^T$, $\textbf{q}=[x_n, y_n, 0]^T$, and $\textbf{v}=[x_b, y_b, 0]^T$, respectively. Moreover, we assume the availability of precise channel state information at the AP through pilot signals, establishing a maximum achievable performance level for this practical system \cite{8811733}. To maximize the effectiveness of the MEC server, UEs adopt a partial offloading scheme for task placement. Under this scheme, tasks are partitioned into two components: one portion is processed locally, while the other is offloaded to the MEC server for more intensive computation. In the following, we first discuss the communications model of the considered system. Then, we provide the task computational and energy consumption models, followed by the problem formulation.
\subsection{Communication Model}
\label{Communication Model}
Let $U_n=\{S_n, C^{'}_n\}$ represent the $n$-th UE task, where $S_n$ and $C^{'}_n$ denote the size (in bits) and computational cycle (cycles/bit) requirements, respectively. Moreover, the number of cycles requested to compute $S_n$ of bits is expressed as $C_n=C^{'}_nS_n$. To handle extensive computation, $n$-th UE offloads a portion of their tasks to the MEC system through the communication link $\textbf{h}_n$, which can be expressed as:
   ${\textbf{h}}_n= \textbf{h}_{n,b}+\textbf{G}\Phi \textbf{h}_{n,u}.$
Where, $\textbf{h}_{n,b}=\sqrt{PL_{n,b}}\Hat{\mathbf{h}_{t,n,b}}\in \mathbb{C}^{M \times 1}$
represents the communication link between the UE $n$ and the MEC server, $\textbf{G}=\sqrt{PL_{b,u}}\Hat{\mathbf{h}}_{r,n,b}\Hat{\mathbf{h}}^H_{t,n,b} \in \mathbb{C}^{M \times K}$ represents the communication link between the UEs and the BD-IRS-UAV, and $\textbf{h}_{n,u}=\sqrt{PL_{n,u}}\Hat{\mathbf{h}_{t,n,u}} \in \mathbb{C}^{K \times 1}$ represents the communication link between the UAV-IRS and the MEC server. Where $PL_{(.)}$, $\Hat{\mathbf{h}}_{r,(.),(.)}$ and $\Hat{\mathbf{h}}_{t,(.),(.)}$ represent path loss and the antenna response vectors that can be expressed mathematically as\cite{9651523}:
\begin{equation}
\small
\label{UE-UAV}
\Hat{\mathbf{h}}_{(.),(.),(.)}=\left[1,e^{\frac{-j2\pi d_{o}}{\lambda}\vartheta_{(.),(.)}},\dots,e^{\frac{-j(I-1)2\pi d_{o} }{\lambda}\vartheta_{(.),(.)}}\right] ^T.
\end{equation}
In \eqref{UE-UAV} $\vartheta_{(.),(.)}$ represents the AOA/AOD of the signal, $I\in \{K,M\}$,  $d_o$ represents the separation distance, and $\lambda$ denotes the wavelength of the carrier. Moreover, the reflective phase shift coefficient matrix is denoted as $\boldsymbol{\Phi} \in \mathbb{C}^{K \times K}$. The number of non-zero elements in $\boldsymbol{\Phi}$ depends on the architecture of the IRS, which follows different IRS element connection topologies as discussed in the following subsections. The IRS phase shift is configured by the AP and transmitted to the IRS controller through a dedicated control link\cite{9651523}.
\subsubsection{Single-Connected IRS} In the single-connected IRS, elements are not connected to each other, resulting in the diagonal the phase shift coefficient matrix:
\begin{equation}
\boldsymbol{\Phi}=\text{diag}\{\phi_1,\phi_2,\cdots, \phi_K\},\; \text{s.t.}\; |\phi_k|^2=1, \forall k.
\end{equation}
\subsubsection{Fully-Connected IRS} In this architecture, IRS elements are fully connected to each other through re-configurable impedance components, resulting in a full matrix satisfying the following constraint:
$\boldsymbol{\Phi}^H\boldsymbol{\Phi}=\textbf{I}_K.$
\subsubsection{Group-Connected IRS} In the group connected IRS architecture, a total of $K$ elements are divided into $L$ groups denoted by $\mathcal{L}=\{l|1,2,3,\cdots L\}$, where each group follows the fully connected topology. Each group has the same number of elements for convenience, denoted by $\Bar{K}=K/L$. Based on this, $\boldsymbol{\Phi}$ is a block diagonal matrix given by:
\begin{equation}
\boldsymbol{\Phi}= \text{blkdiag}\{\boldsymbol{\Phi}_1,\boldsymbol{\Phi}_2,\cdots, \boldsymbol{\Phi}_L\},
\end{equation}
where each entity $\boldsymbol{\Phi}\in \mathbb{C}^{\Bar{K} \times \Bar{K}}$  satisfy:
$\boldsymbol{\Phi}_l^H\boldsymbol{\Phi}_l=\textbf{I}_{\Bar{K}},\forall l \in \mathcal{L}.$
\begin{table}
\centering
\caption{RIS architectures \& non-zero element count}
\label{tab:1}
\resizebox{\columnwidth}{!}{%
\begin{tabular}{cccccl|}
\cline{2-6}
\multicolumn{1}{c|}{} &
  \multicolumn{1}{c|}{No. G} &
  \multicolumn{1}{c|}{GD} &
  \multicolumn{1}{c|}{E/G} &
  \multicolumn{1}{c|}{No. NZE} &
  Constraints \\ \hline
\multicolumn{1}{|c|}{SC} &
  \multicolumn{1}{c|}{$K$} &
  \multicolumn{1}{c|}{$1$} &
  \multicolumn{1}{c|}{$1$} &
  \multicolumn{1}{c|}{$K$} &
  $|\phi_k|^2=1, \forall k\in \mathcal{K}$ \\ \hline
\multicolumn{1}{|c|}{FC} &
  \multicolumn{1}{c|}{1} &
  \multicolumn{1}{c|}{$K$} &
  \multicolumn{1}{c|}{$K^2$} &
  \multicolumn{1}{c|}{$K^2$} &
  $\boldsymbol{\Phi}^H\boldsymbol{\Phi}=\textbf{I}_K.$ \\ \hline
\multicolumn{1}{|c|}{GC} &
  \multicolumn{1}{c|}{$L$} &
  \multicolumn{1}{c|}{$\Bar{K}$} &
  \multicolumn{1}{c|}{$\Bar{K}^2$} &
  \multicolumn{1}{c|}{$L\Bar{K}^2$} &
  $\boldsymbol{\Phi}_l^H\boldsymbol{\Phi}_l=\textbf{I}_{\Bar{K}}, \quad \forall l \in \mathcal{L}.$ \\ \hline
\multicolumn{6}{|l|}{\begin{tabular}[c]{@{}l@{}}SC: Single Connected, FC, Fully Connected, GC, Group Connected, \\ GD: Group Dimension, E/G: Element per Group, NZE; Non-Zero Elements\end{tabular}} \\ \hline
\vspace{-5mm}
\end{tabular}%
}
\end{table}
\par
Table \ref{tab:1} provides an overview of IRS architectures and element topologies with associated constraints. Each entity of the $\boldsymbol{\Phi}$ is represented as $\phi_k = e^{j\theta_k}$, where $\theta_k$ is a phase shift in the range of $[0, 2\pi]$, and $\xi \in [0, 1]$ represents the amplitude of the reflection coefficient  $k$-th element. In this work, we only consider IRS with full reflection, and as such, we set $\xi = 1$. 
\par 
Thus, the achievable data rate for offloading of the $n$-th UE task is given by $R_n = B\log_2\left(1+\gamma_n\right)$, where $\gamma_n$ represents the signal-to-interference and noise (SINR) ratio and can be mathematically expressed as follows:
\begin{equation}
    \label{SINR}
    \gamma_n\!=\!\frac{p_n|\boldsymbol{w}_n^H\left(\textbf{h}_{n,b}+\textbf{G}\boldsymbol{\Phi} \textbf{h}_{n,u}\right)|^2}{\sum_{n^{'}\ne n}^Np_{n^{'}}|\boldsymbol{w}_{n}^H\left(\textbf{h}_{{n^{'}},b}+\textbf{G}\boldsymbol{\Phi} \textbf{h}_{{n'},u}\right)|^2\!\!+\!{\sigma}^2|\boldsymbol{w}_n^H|^2}.
\end{equation}
Here, $\boldsymbol{w}_n \in \mathbb{C}^{M \times 1}$ represents the receive beamforming vector at the AP, and $p_n$ denotes the transmission power of the $n$-th UE. Additionally, ${\sigma^2} $ represents the additive white Gaussian noise.
\subsection{Task Computational and Energy Consumption Model}
We consider a partial offloading scheme,  which facilitates parallel computation of resource-intensive tasks at both the MEC and the devices. The extent to which bits are computed locally is denoted by the variable $\beta_n \in [0,1]$, in which $\beta_n$ parts of the task are computed locally, and $1-\beta_n$ parts of the task will be off-loaded to the MEC server.

\subsubsection{Local Computation}
In the local computational scheme, UEs utilize their local computational resources $f_n^l \in [0,1]$ to compute the task. As such, the time required to compute the task locally can be expressed as follows:
\begin{equation}
    \label{local time}
    t_n^l=\frac{\beta_nC_n}{f_n^l f_n^{max}}.
\end{equation}
In equation \eqref{local time}, $f_n^{max}$ represents the maximum computation resources of the $n$-th UE. Similarly, the energy consumed during the local computation of the task is given by:
\begin{equation}
    \label{local Energy}
E_n^l=\epsilon\left(f_n^lf_n^{max}\right)^2\beta_nC_n.
\end{equation}
Here, $\epsilon$ represents the computational energy efficiency, which is assumed to be the same for all devices.
\subsubsection{Edge Computation}
UEs offload the portions of the task to the MEC server for extensive and parallel computation using a partial offloading scheme. The total computational time of the task at the MEC server is a combination of both the offloading time and computational time. As per the communication model discussed in Section \ref{Communication Model}, the offloading time is expressed as:
\begin{equation}
\label{offloading time}
t_n^o=\frac{(1-\beta_n)S_n}{R_n} .
\end{equation}
Similarly, the energy required for task offloading can be calculated as $E_n^o = p_n t_n^o$. Upon receiving the computation task at the AP, the MEC allocates a portion of its shared computational resources, denoted by $f_n^s \in [0,1]$, to perform the intensive computation. Similarly, the time consumed while computing the task at the MEC can be expressed as: 
\begin{equation}
    \label{Edge time}
    t_n^s=\frac{(1-\beta_n)C_n}{f_n^sf_s^{max}}.
\end{equation}
In addition to the maximum computational resources available at the MEC represented by $f_s^{max}$, the time it takes to send the task back from the MEC server to the UE is not considered in this work due to its negligible size. The total computational time for edge computing is expressed as $t_n^E=t_n^o+t_n^s$. We also assume that the AP has an infinite and continuous power supply, and this work does not account for the energy consumption required to perform the MEC's extensive computation task. Therefore, the analysis considers only the energy consumed for offloading the task portion to the MEC and the energy consumed for local computation \cite{mahmood2021optimal}.
\subsubsection{UAV Energy Consumption Model}
Since the UAV passively reflects the signal through {\color{black}BD-IRS}, there is no energy consumed for communications. The UAV's energy consumption is for hovering to relay the signal, mainly determined by the off-loading time and the computation time at the edge, as the time for returning the result is negligible.  Consequently, the overall UAV hover time can be expressed as $t^h = \max\{t_1^E, t_2^E, \cdots, t_N^E\}$, and the energy consumption of the UAV is denoted as $E_u^H = P^h t^h$, where $P^h$ signifies the constant hovering power, contingent on the UAV's hardware configuration.
\subsection{Problem Formulation}
In this study, task $U_n$ is computed in parallel at the MEC server and locally using a partial offloading scheme. As a result, the computational delay for the task and UAV hovering time is represented as $T_n=\max(t_n^l,t_n^E)$ and $t^h = \max\{t_1^E, t_2^E, \cdots, t_N^E\}$ respectively, where $t_n^E=t_n^o+t_n^s$ denotes the computing delay experienced at the MEC server. Our objective is to minimize the utility function, taking into account the UAV's hovering energy consumption and task computational delay,  by jointly optimizing the computational resources of MEC $\boldsymbol{f}^s=[f_1^s,f_2^s\cdots,f_N^s]$ and UEs $\boldsymbol{f}^l=[f_1^l,f_2^l,\cdots,f_N^l]$, transmission power $\boldsymbol{p}=[p_1,p_2,\cdots,p_N]$, task segmentation variable $\boldsymbol{\beta}=[\beta_1,\beta_2,\cdots,\beta_N]$, active beam-forming vector at AP $\boldsymbol{W}=[\boldsymbol{w}_1,\boldsymbol{w}_2,\cdots,\boldsymbol{w}_N]$, the phase shift $\boldsymbol{\Phi}$, and location of UAV $\textbf{z}$. We can mathematically formulate the joint optimization problem as follows:
\allowdisplaybreaks
\begin{subequations}
\label{P0}
\begin{align}
\label{OF0}
\!\!\mathcal{P}_0: \!\!\!\!\min_{{\substack{\boldsymbol{f}^s,\boldsymbol{f}^l,\boldsymbol{p}\\\boldsymbol{\beta},\boldsymbol{W},\boldsymbol{\Phi},\textbf{z}}}}\!\! & w_1\!P^h\max_{n}\left(t_n^E\right)+w_2\!\max_{n}\left(T_n\right)\\
\label{P0C01}
\textrm{s.t.}\quad &{\sum}_{n=1}^{N}f_n^s\le 1,\\
\label{P0C02}
\quad &\gamma_n\ge\gamma_n^{min},\forall n\in \mathcal{N},\\
\label{P0C03}
\quad &E_n^l+E_n^o+P_rt_n^o\le E^{max}, \forall n\in \mathcal{N},\\
\label{P0C04}
\quad &d_u^b=\norm{\textbf{z}-\textbf{v}}^2\le d_{u,b}^{max},\\
\label{P0C05}
\quad &d_n^u=\norm{\textbf{q}-\textbf{z}}^2\le d_{n,u}^{max},\forall n\in \mathcal{N},\\
\label{P0C06}
\quad &\norm{\boldsymbol{w}_n}^2=1,\forall n\in \mathcal{N},\\
\label{P0C07}
\quad& 0 \le \beta_n\le 1, \forall n\in \mathcal{N},\\
\label{P0C08}
\quad & \boldsymbol{\Phi}\quad\text{as per Table \ref{tab:1}} \\
\label{P0C09}
\quad & f_n^l\ge0, f_n^E\ge0, p_n\ge0, \forall n\in \mathcal{N}.
\end{align}
\end{subequations}
Several constraints were imposed to achieve the objective in \eqref{P0}. Constraint \eqref{P0C01} ensures that the MEC's shared computation resources do not exceed the maximum available resources. Imposing SINR requirements for reliable communication, constraint \eqref{P0C02} is crucial. Furthermore, constraint \eqref{P0C03} limits UEs' energy consumption, restricting it at a maximum threshold of $E^{max}$, where $P_r$ represents each UE's constant energy consumption. Constraints \eqref{P0C04} and \eqref{P0C05} play a vital role in guaranteeing communication link quality and imposing maximum distance limits between the UAV-AP and UE-UAV, set as $d_{u,b}^{max}$ and $d_{n,u}^{max}$, respectively.
Constraint \eqref{P0C06} enforces the beamforming constraint at the AP, while constraint \ref{P0C07} establishes bounds for the task segmentation variable. Ensuring the unit modulus constraint of the IRS, constraint \ref{P0C08} is essential, and finally, constraint \eqref{P0C09} specifies the bounds for the decision variables.
\begin{remark}
The optimization problem $\mathcal{P}_0$ presents a non-linear and non-convex nature, primarily due to $\log$ function in $t_n^o$ within equation \eqref{OF0}. Additionally, the coupling between beamforming vector $\boldsymbol{w}$, phase shift $\boldsymbol{\Phi}$, and transmission power $\boldsymbol{p}$ renders $\mathcal{P}_0$ an NP-hard problem. As a result, finding the global optimal solution is a challenging task. To address this issue, we propose a solution based on block coordinate descent (BCD) methods that decouple the original optimization problem into a series of subproblems that can be solved iteratively to obtain a locally optimal solution for $\mathcal{P}_0$ at low complexity.
\end{remark}
\section{A Framework for Joint Optimization of Computing and Communication Resources}
\label{Proposed Solution}
This section introduces an efficient framework for optimizing BD-IRS-UAV placement, communication, and computational resources using the BCD technique, which iteratively optimizes one variable while keeping the others fixed until convergence. In the subsequent section, we optimize computational resources and BD-IRS-UAV location while keeping communication resources constant. We then proceed to optimize communication resources while holding computational resources fixed. This approach effectively manages the joint optimization problem, leading to convergence toward a locally optimal solution. Although it does not guarantee the global optimal solution, the performance gap to achieve the global optimal solution is left for future work. In the following, we will detail the sub-problems, whose overview diagram of the problem decomposition is shown in Fig.~\ref{fig:RoadMap}.

\begin{figure}[ht]
	\centering
	\includegraphics[width=0.90\linewidth]{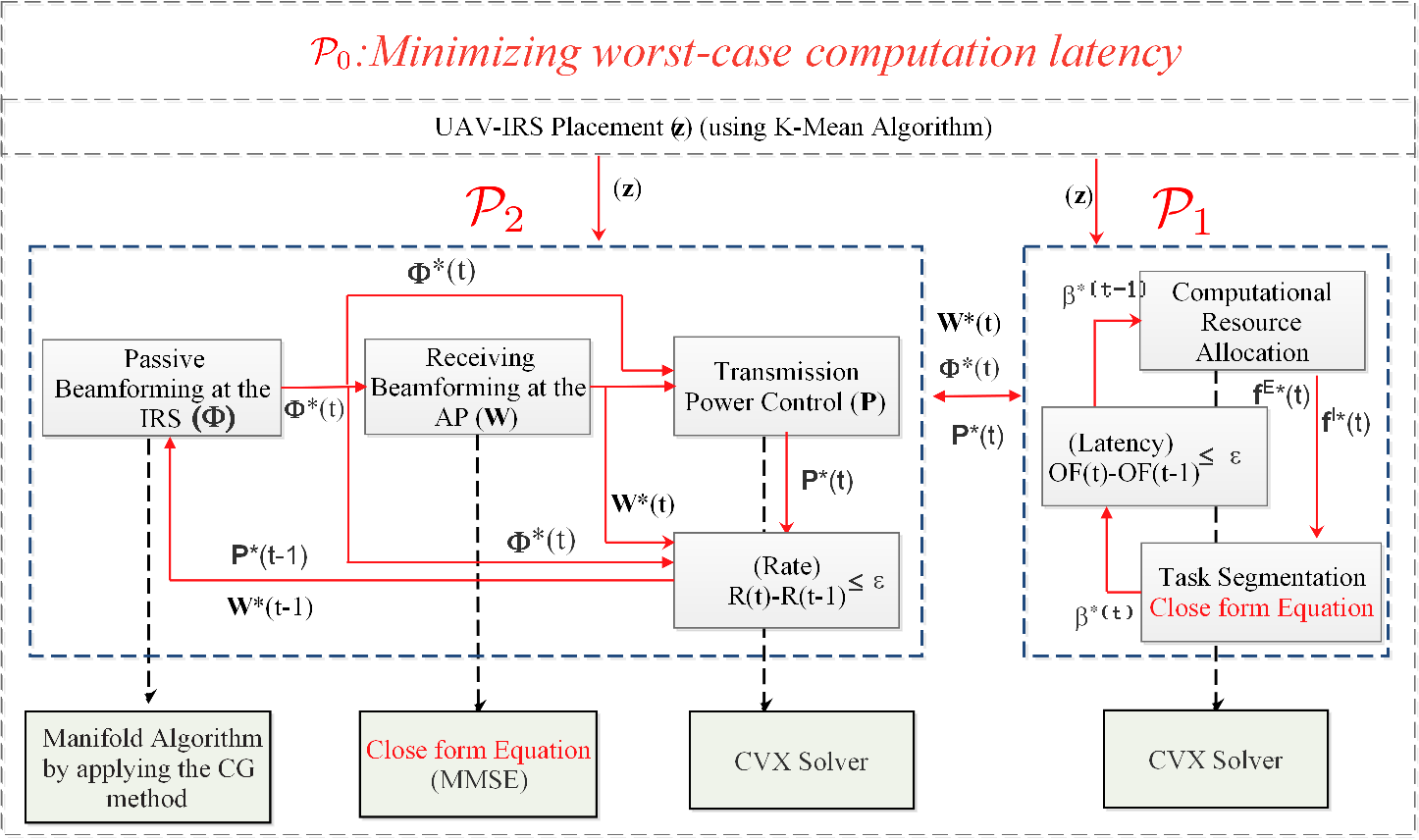}
	\caption{Road-map of iterative algorithms.}
	\label{fig:RoadMap}
	\vspace{-0.3cm}
\end{figure}

\subsection{UAV Placement and Computational Resource Optimization
}
\label{OCS}
This section aims to determine the optimal values of computational resources ($\boldsymbol{f^s, f^l}$), task segmentation variable ($\boldsymbol{\beta}$) and {\color{black}BD-IRS-UAV}  placement ($\textbf{z}$) that minimize the \eqref{OF0}, under given the communication resources ($\boldsymbol{w, \Phi, p}$). Notably, equation \eqref{OF0} indicates the task's computational time influenced by the offloading and edge computing times, directly impacting the UAV's hovering time. Thus, our goal of minimizing the UAV's hovering time is equivalent to minimizing the task's computational time. Building upon this insight, we reformulate the objective function as follows:
\allowdisplaybreaks
\begin{subequations}
\label{P1}
\begin{align}
\label{OF1}
\mathcal{P}_1: &\min_{{\substack{\boldsymbol{f}^s,\boldsymbol{f}^l,\boldsymbol{\beta},\textbf{z}}}}  \max_{n}\left(T_n\right)\\
\!\!\!\!\!\!\!\!\!\!\!\!&\text{\cref{P0C01,P0C03,P0C04,P0C05,P0C07,P0C09}}.\notag
\end{align}
\end{subequations}
Jointly optimizing ($\boldsymbol{f^s, f^l}$) and $\boldsymbol{\beta}$ is challenging. To address this, we initiated the placement of BD-IRS-UAV, intending to strategically position the UAV in a 3D plane, thereby enhancing the link quality for each user by minimizing the distance, resulting in improved objective function values. To determine the UAV location, a practical and effective heuristic approach based on a modified K-means algorithm is designed to work sequentially and be distributed as given in \cite{el2019learn}. Unlike traditional K-Means algorithms, this modified version situates the UAV at the barycenter of all users, not just the closed users. In other words, it computes the UAV's position as $x_u\leftarrow\frac{1}{|N|}{\sum}_{n\in \mathcal{N}}x_n$ and $y_u\leftarrow\frac{1}{|N|}{\sum}_{n\in \mathcal{N}}y_n$, where the height of the UAV is optimized in the range $H \in [H_{\min}, H_{\max}] $ that minimizes the path loss that can be expressed mathematically as 
\begin{equation}
H^* = \underset{H \in [H_{\min}, H_{\max}]}{\arg\min} \, PL(H),
\end{equation}
while ensuring \eqref{P0C04} and \eqref{P0C05} and resulting improved link quality. Moreover, as per \eqref{OF1}, ${\beta}_n$ primarily depends on computational resources. Moreover, under the given values of $\boldsymbol{f^s}$ and $\boldsymbol{f^l}$, the ${\beta}_n$ directly affects the local and edge computing. So, to find the closed-form equation for ${\beta}_n$, we can look at the monotonic relationship between offloading and local computational time, which is shown in Lemma \ref{lemma-beta}.
\begin{lemma}
\label{lemma-beta}
Under the feasible value of $\boldsymbol{f^s\text{~and~}f^l}$ the value of task segmentation value is given by:
\begin{equation}
\small
    \label{tasksegmentation} 
    \beta_n=\frac{f_n^lf_n^{max}\left(R_nC_n+f_n^sf_s^{max}S_n\right)}{R_nC_n\left(f_n^lf_n^{max}+f_n^sf_s^{max}\right)+S_nf_n^lf_n^{max}f_n^sf_s^{max}},
\end{equation}
whereas the optimal values $\boldsymbol{f^{s*}\text{~and~}f^{s*}}$ results in the optimal value of  $\boldsymbol{\beta}^*$.
\end{lemma}
\begin{proof}
See Appendix \ref{APA}
\end{proof}
\subsubsection{Local and Edge Computational Resources}
Given the value of task segmentation available, the optimization problem for optimal allocation of local and edge computational resources can be expressed as follows:
\allowdisplaybreaks
\begin{subequations}
\label{P1A}
\begin{align}
\label{OF1A}
\mathcal{P}_{1A}: &\min_{{\substack{\boldsymbol{f}^s,\boldsymbol{f}^l}}}  \max_{n}\left(T_n\right)\\
&\text{\cref{P0C01,P0C03,P0C09}}.
\end{align}
\end{subequations}
Since  $\mathcal{P}_{1A}$ is convex, it can be solved efficiently using the convex optimization toolbox, e.g., CVX. Whereas the procedure for solving the $\mathcal{P}_{1A}$ is expressed in Algorithm \ref{algo:Algorithem 1}.

\begin{algorithm2e}[t]
\small
\SetAlgoLined
\textbf{Inputs:  $\textbf{h}_n^b$, $\textbf{h}_n^u$, $\textbf{G}$,$f_n^{max}$, $f_s^{max}$, $C_n$, $S_n$, $\boldsymbol{p}$ 
 $\boldsymbol{w}$, $\boldsymbol{\Phi}$, $\textbf{z}$, $\epsilon$}\;
\textbf{Output: $\boldsymbol{\beta}^*$, $\boldsymbol{{f^s}^*,{f^l}^*}$}\;
\textbf{Initialization:} ${\boldsymbol{f}_{t}^s},{\boldsymbol{f}_{t}^l}$ \text{satisfying} \eqref{P0C01} and \eqref{P0C03}\;
\While{error $>$ $\epsilon$ \&\& $i< i^{max}$}{
        Calculate $\boldsymbol{\beta}_{i+1}$ using  \eqref{tasksegmentation}\\
        Calculate ${\boldsymbol{f_{i+1}^s},\boldsymbol{f_{i+1}^l}}$ using \eqref{P1A}\\
       //Calculate Objective function value by solving \eqref{OF1} and compute error//.\\
       error=$\frac{|\text{obj}\left({\boldsymbol{f_{i+1}^s},\boldsymbol{f_{i+1}^l}},\boldsymbol{\beta}_{i+1}\right)-\text{obj}\left({\boldsymbol{f_{i}^s},\boldsymbol{f_{i}^l}},\boldsymbol{\beta}_{i}\right)|}{\text{obj}\left({\boldsymbol{f_{i+1}^s},\boldsymbol{f_{i+1}^l}},\boldsymbol{\beta}_{i+1}\right)}$\\
       $i\leftarrow i+1 $ // Update value of i\;
        
}
\caption{\textbf{{Iterative \!Algorithm to\! solve\!} \eqref{P1}}}\label{algo:Algorithem 1}
\end{algorithm2e}
\subsubsection{Convergence and  Complexity Analysis of Algorithm \ref{algo:Algorithem 1}}
\paragraph{Complexity}
This section provides a worst-case per-iteration complexity analysis for Algorithm \ref{algo:Algorithem 1}, designed to iteratively solve \eqref{P1A} and \eqref{tasksegmentation}.  Specifically, problem \eqref{P1A} encompasses $2N$ decision variables and $3N+1$ constraints resulting per iteration computational complexity as  $\mathcal{O}\left((2N)^3(3N\!+\!1)\right)$. Furthermore, since \eqref{tasksegmentation} solely relies on parameter $N$, the per-iteration complexity to solve \eqref{tasksegmentation} characterized as $\mathcal{O}(N)$. In addition, for $i^{max}$ iterations, the overall complexity of Algorithm \ref{algo:Algorithem 1} is $\small{\mathcal{O}\left(i^{max}\left((2N)^3(3N+1)+N\right)\right)}$.
\paragraph{Convergence}  Let $\mathcal{Q}=[\boldsymbol{f}_i^{s*},\boldsymbol{f}_i^{l*}]$ denote the optimal solution of \eqref{P1A}, resulting in $\boldsymbol{\beta}i^*$ at current iteration. The objective function value of \eqref{P1A} at this iteration is represented as $\mathbb{F}(\mathcal{Q}_i)$ corresponds to a subset of the original optimization problem. Furthermore, the objective function value iteratively improves, i.e., $\mathbb{F}(\mathcal{Q})_{i+1})\le\mathbb{F}(\mathcal{Q}_{i})$. According to Beck et al. \cite{beck2010sequential}, a stable point is reached when the difference between current and previous iterations is less than $\epsilon$, i.e., $\mathbb{F}{i+1}-\mathbb{F}{i}\le \epsilon$.

\subsection{Optimizing Communication Resources}
Given the {\color{black}BD-IRS-UAV} location, task segmentation vector, and computational resources vector, the optimization problem for optimizing  communication resources can be expressed as: 
\allowdisplaybreaks
\allowdisplaybreaks
\begin{subequations}
\label{P2}
\begin{align}
\label{OF2}
\mathcal{P}_2: &\min_{{\substack{\boldsymbol{p},\boldsymbol{w},\boldsymbol{\Phi}}}}\max_{n}\left(T_n\right)\\
&\text{\cref{P0C02,P0C03,P0C06,P0C08,P0C09}}
\end{align}
\end{subequations}
\begin{remark}
The optimization problem \eqref{P2} remains challenging to solve due to its non-convex and non-linear nature, resulting from the fractional objective function and coupling of decision variables $\boldsymbol{p,w, \Phi}$. Moreover, the interdependent nature of these variables further complicates the problem, classifying it as NP-hard. To address these challenges, our approach, detailed in the subsequent section, entails a two-fold transformation strategy. Initially, we reconfigure the objective function into a more tractable form, as described in Section \ref{T1}. Subsequently, we employ a re-transformation process, detailed in Section \ref{T2}, to overcome its non-convex and non-linear nature and find an efficient solution.
\end{remark}
\paragraph{Problem Transformation}
\label{T1}
Lemma \ref{lemma-beta} shows that the optimal solution of problem $\mathcal{P}_0$ is such that $t_n^l=t_n^E$. By replacing $t_n^l$ with $t_n^E$ and removing the constant term, we can reformulate $\mathcal{P}_2$ as follows:
\allowdisplaybreaks
\begin{subequations}
\label{P2A}
\begin{align}
\label{OF2A}
\mathcal{P}_{2A}: &\min_{{\substack{\boldsymbol{p},\boldsymbol{w},\boldsymbol{\Phi}}}}\max_{n}\left(\frac{\beta_nS_n}{R_n}\right)\\
&\text{\cref{P0C02,P0C03,P0C06,P0C08,P0C09}}.\notag
\end{align}
\end{subequations}
Following that, to overcome the fractional nature of the objective function, we introduce an auxiliary variable $\Gamma_n$ and reformulate them as follows:
\allowdisplaybreaks
\begin{subequations}
\label{P2A1}
\begin{align}
\label{OF2A1}
\mathcal{P}_{2A1}: &\min_{{\substack{\boldsymbol{p},\boldsymbol{w},\boldsymbol{\Phi},\boldsymbol{\Gamma}}}}\max_{n}\left(\Gamma_n\right),\;\text{s.t.}\;\frac{\beta_nS_n}{R_n}\le \Gamma_n, \forall n \in \mathcal{N}.\\
&\text{\cref{P0C02,P0C03,P0C06,P0C08,P0C09}}.\notag
\end{align}
\end{subequations}
The following Lemma aids in optimizing \ref{OF2A1}:
\begin{lemma}
\label{lemma-3}
If $\left(\{\boldsymbol{p}^*\},\{\boldsymbol{w}^*\},\{\boldsymbol{\Phi}^*\},\{\boldsymbol{\Gamma}^*\}\right)$ is the optimal solution for  \ref{P2A1}, then there must exist a $\{\boldsymbol{\chi}^*\}$ such that $\left(\{\boldsymbol{p}^*\},\{\boldsymbol{w}^*\},\{\boldsymbol{\Phi}^*\}\right)$ satisfies the Karush-Kuhn-Tucker condition of the following problem for $\boldsymbol{\chi}=\boldsymbol{\chi}^*$ and $\boldsymbol{\Gamma}=\boldsymbol{\Gamma}^*$
\begin{subequations}
\label{P2A2}
\begin{align}
\label{OF2A2}
\mathcal{P}_{2A2}: &\min_{{\substack{\boldsymbol{p},\boldsymbol{w},\boldsymbol{\Phi},\textbf{z}}}}\max_{n}\left(\chi_n(\beta_nS_n-\Gamma_nR_n)\right),\\
&\text{\cref{P0C02,P0C03,P0C06,P0C08,P0C09}}.\notag
\end{align}
\end{subequations}
Furthermore, it also satisfies the following constraints
\begin{equation}
\label{KKTConstraint}
\begin{aligned}
\chi_n=\frac{1}{R_n^*}    , \;
\Gamma_n=\frac{\beta_nS_n}{R_n^*}    , \forall n \in \mathcal{N}.
\end{aligned}
\end{equation}
If $\left({\boldsymbol{p}^*},{\boldsymbol{w}^*},{\boldsymbol{\Phi}^*}\right)$ is a solution to problem $\mathcal{P}_{2A2}$ that satisfies the constraint \eqref{KKTConstraint} by setting $\boldsymbol{\chi}=\boldsymbol{\chi}^*$ and $\boldsymbol{\Gamma}=\boldsymbol{\Gamma}^*$, then it follows that $\left({\boldsymbol{p}^*},{\boldsymbol{w}^*},{\boldsymbol{\Phi}^*},{\boldsymbol{\Gamma}^*}\right)$ is a solution to problem $\mathcal{P}_{2A1}$ for the Lagrange variable $\boldsymbol{\chi}=\boldsymbol{\chi}^*$.
\end{lemma}
\begin{proof}
    Please refer to Section 2 of \cite{jong2012efficient}
\end{proof}
To this end, the fractional optimization problem $\mathcal{P}_2$ transformed into a more trackable form $\mathcal{P}_{2A2}$, which is going to be solved in two steps. In the first step, we will find the $\left({\boldsymbol{p}^*},{\boldsymbol{w}^*},{\boldsymbol{\Phi}^*}\right)$ by solving   $\mathcal{P}_{2A2}$, whereas in the second step, we will update the value of ${\boldsymbol{\Gamma}} ~\text{and}~{\boldsymbol{\chi}}$.
\subsubsection{Step 1}
\label{T2}
Under the given value of  ${\boldsymbol{\Gamma}} ~\text{and}~{\boldsymbol{\chi}}$, the communication resources optimization problem can be simplified to:
\begin{subequations}
\small
\label{P2A3}
\begin{align}
\label{OF2A3}
\mathcal{P}_{2A3}: &\max_{{\substack{\boldsymbol{p},\boldsymbol{w},\boldsymbol{\Phi}}}}\min_{n}\left(\chi_n\Gamma_nR_n\right),\\
&\text{\cref{P0C02,P0C03,P0C06,P0C08,P0C09}}
\end{align}
\end{subequations}
The optimization problem \eqref{P2A3} presents a level of tractability; however, it retains its non-convex and non-linear nature, mainly due to the logarithmic function in  \eqref{OF2A3}. To overcome this challenge, we adopt techniques based on Lagrange dual transformation. The rate function in \eqref{P2A3} define as $R_n\left({p}_n,\boldsymbol{w}_n,\boldsymbol{\Phi}\right)=\log_2\left(1+\gamma_n\right)$, with $\gamma_n$ denoting the fractional SINR term for the $n$-th user. In order to address the non-convex and non-linear nature of the objective function, we introduce auxiliary variables $\psi_n \in \mathbb{R}^{N}$ corresponding to each $\gamma_n$ term and express the new objective function as follows: 
\begin{equation}
\small
\label{O1}
\begin{aligned}
R_n\left({p}_n,\boldsymbol{w}_n,\boldsymbol{\Phi},{\psi}_n\right)=\Bigl(\log_2\left(1+\psi_n\right)-\psi_n \\ 
+\frac{\left(1+\psi_n\right)p_n|\boldsymbol{w}_n^H {\textbf{h}}_n|^2}{\sum_{n^{'}\ne n}^Np_{n^{'}}|\boldsymbol{w}_{n}^H {\textbf{h}}_{n'}|^2\!\!+\!\boldsymbol{\sigma}^2|\boldsymbol{w}_n^H|^2}\biggr).
\end{aligned}
\end{equation}
In \eqref{O1}, the fractional term of SINR is independent of the logarithmic function, making it more straightforward to analyze. To facilitate this analysis, we employ a quadratic transformation technique to convert the fractional component of equation \eqref{O1} into an integer expression. Consequently, the modified objective function can be expressed as follows:
\begin{equation}
  \begin{aligned}
&R_n\left({p}_n,\boldsymbol{w}_n,\boldsymbol{\Phi},{\psi}_n,{\zeta}_n\right)=\Bigl(\log_2\left(1+\psi_n\right)-\psi_n\\
&+2\sqrt{(1+\psi_n)p_n}\mathcal{R}\{\zeta_n^{\dagger}\boldsymbol{w}_{n}^H {\textbf{h}}_{n}\}\\
&-|\zeta_n|^2\left({\sum}_{{n'}\ne n}^Np_{n'}|\boldsymbol{w}_{n}^H {\textbf{h}}_{n'}|^2+\boldsymbol{\sigma}^2|\boldsymbol{w}_n^H|^2\right)\biggr).
  \end{aligned}
\end{equation}
The variable $\zeta_n\in \mathbb{C}^{N}$ represents an additional auxiliary variable. With the introduction of these two auxiliary variables, we can reformulate the optimization problem \eqref{P2A3} as follows:
\begin{subequations}
\label{p3b}
\begin{align}
\mathcal{P}_{2A4}:  &\!\!\!\!\max_{{\substack{\boldsymbol{p},\boldsymbol{w},\boldsymbol{\Phi},\boldsymbol{\psi},\boldsymbol{\zeta}}}}\quad \!\!\!\!\!\!\min_n\left(\chi_n\Gamma_nR_n\left({p}_n,\boldsymbol{w}_n,\boldsymbol{\Phi},{\psi}_n,{\zeta}_n\right)\!\right)\\
&\text{\cref{P0C02,P0C03,P0C06,P0C08,P0C09}}.\notag
\end{align}
\end{subequations}
The problem presented in equation \eqref{p3b} is an optimization problem with multiple variables. One well-established method for solving such problems is the BCD iterative algorithm. In the following sections, we will decompose the problem into sub-problems and solve them iteratively. The subsequent sections will provide a detailed procedure for finding the solution to these sub-problems.
\paragraph{Update Auxiliary Variables}
Under the given values $\boldsymbol{p},\boldsymbol{w},\boldsymbol{\Phi}$ {and} $\boldsymbol{\psi} (\text{or} \boldsymbol{\zeta})$,  \eqref{p3b} is unconstrained optimization problem and thus solved by  
$
\frac{\partial R\left(\boldsymbol{\Phi},\boldsymbol{w},\boldsymbol{p},\boldsymbol{\psi},\boldsymbol{\zeta}\right)}{\partial \boldsymbol{\psi}}\!\!=\!\!0
$
or
$\!
\frac{\partial R\left(\boldsymbol{\Phi},\boldsymbol{w},\boldsymbol{p},\boldsymbol{\psi},\boldsymbol{\zeta}\right)}{\partial \boldsymbol{\zeta}}=0
$
and obtain the solution of each auxiliary variable as follows;
\begin{equation}
\label{AUX1}
    \psi_n=\frac{\left(1+\psi_n\right)p_n|\boldsymbol{w}_n^H {\textbf{h}}_n|^2}{\sum_{n^{'}\ne n}^Np_{n^{'}}|\boldsymbol{w}_{n}^H {\textbf{h}}_{n'}|^2\!\!+\!\boldsymbol{\sigma}^2|\boldsymbol{w}_n^H|^2}
\end{equation}

\begin{equation}
\small
\label{AUX2}
    \zeta_n=\frac{\sqrt{(1+\psi_n)p_n}\boldsymbol{w}_n^H {\textbf{h}}_n}{\sum_{n^{'}\ne n}^Np_{n^{'}}|\boldsymbol{w}_{n}^H {\textbf{h}}_{n'}|^2\!\!+\!\boldsymbol{\sigma}^2|\boldsymbol{w}_n^H|^2}
\end{equation}
From equations \eqref{AUX1} and \eqref{AUX2}, it can be observed that the value of the auxiliary variables is dependent on the transmission power $p_n$, phase shift $\boldsymbol{\Phi}$,  beamforming vector $\textbf{w}_n$ at the AP. Therefore, the optimal values of decision variables will lead to the optimal value of the auxiliary variables. The following sections present an efficient solution to finding these variables.

\paragraph{Beam-forming at AP}
Under the given values of ${\boldsymbol{p},\boldsymbol{\Phi},\boldsymbol{\psi},\boldsymbol{\zeta}}$, the optimization problem for the beamforming vector $\boldsymbol{w}$ at the AP can be expressed as follows:
\begin{subequations}
\label{p2A42}
\begin{align}
\mathcal{P}_{2A4-1}:  \!\!\!\!\max_{{\substack{\boldsymbol{w}}}}\quad &\!\!\!\!\!\!\min_n\left(\chi_n\Gamma_nR_n\left(\boldsymbol{w}_n\right)\!\right), \quad\norm{\boldsymbol{w}_n}^2=1,\forall n.
\end{align}
\end{subequations}
It is evident that both the objective function and the constraints of \eqref{p2A42} exhibit a convex nature. Consequently, it can be inferred that the receiver capable of achieving the maximum rate corresponds to the minimum mean square error-SIC receiver, as described in \cite{Windpassinger2004}
\begin{equation}   
\label{APB}
\!\!\!\!\boldsymbol{w}_n^{'}=\zeta_n\left({\sum}_{n\ne n'}^Np_{n'}\textbf{h}_{n'}^H\textbf{h}_{n'}+\boldsymbol{\sigma}^2\right)^{-1}\!\!\!\!\times\!\! \sqrt{(1+\psi_n)p_n}\textbf{h}_n
\end{equation}
After normalization,  we obtain  $\boldsymbol{w}_n=\frac{ \boldsymbol{w}_n^{'}}{| \boldsymbol{w}_n^{'}|}$.

\paragraph{Power Control at the User}
Given the values of $\boldsymbol{w}_n$, $\boldsymbol{\Phi}$, $\boldsymbol{\psi}$, and $\boldsymbol{\zeta}$, the optimization problem for the transmission power $\boldsymbol{p}$ at the user end can be expressed as follows:
\begin{subequations}
\label{p2A43}
\begin{align}
\!\!\!\!\mathcal{P}_{2A4-2}\!: \!& \max_{{\substack{\boldsymbol{p}}}} \min_n\left(\!\chi_n\Gamma_nR_n\left(\boldsymbol{p}\right)\!\right),\gamma_n\!\ge\!\gamma_n^{\text{min}},  p_n\ge 0,\forall n.
\end{align}
\end{subequations}

It's evident that \eqref{p2A43} represents a linear function with respect to $\boldsymbol{p}$. However, the QoS constraints are nonlinear and nonconvex in nature. Accordingly, the optimization problem can be reformulated as follows:

\begin{subequations}
\label{p2A43A}
\begin{align}
&\mathcal{P}_{2A4-2A}:  \max_{{\substack{\boldsymbol{p}}}}\quad \min_n\left(\chi_n\Gamma_nR_n\left(\boldsymbol{p}\right)\right)\\
&p_n\bar{h}_{n}\ge\gamma_n^{\text{min}}\left({\sum}_{n'\ne n}^Np_{n'}\bar{h}_{n'}+1\right) , \quad p_n\ge 0, \quad \forall n.
\end{align}
\end{subequations}
Here, $\bar{h}_{n}=\frac{\|\boldsymbol{w}_{n}^H\textbf{h}_{n}\|^2}{\boldsymbol{\sigma}^2\|\boldsymbol{w}_n^H\|^2}$.
Therefore, this problem can be efficiently solved by utilizing the convex optimization toolbox.
\paragraph{Phase Shift}
The mathematical representation of the sub-optimization problem concerning the phase shift control of IRS is as follows:
\begin{subequations}
\label{P2A44}
\begin{align}
&\mathcal{P}_{2A4-3}:  \max_{\boldsymbol{\Phi}}\quad \min_n\Bigl(\chi_n\Gamma_nR_n\left(\boldsymbol{\Phi}\right)\Bigr)\\
& \boldsymbol{\Phi}\quad\text{as per Table \ref{tab:1}} .\notag
\end{align}
\end{subequations}
Where $R_n\left(\boldsymbol{\Phi}\right)=
2\sqrt{(1+\psi_n)p_n}\mathcal{R}\{\zeta_n^{\dagger}\boldsymbol{w}_{n}^H {\textbf{h}}_{n}\}-|\zeta_n|^2\left({\sum}_{{n'}\ne n}^Np_{n'}|\boldsymbol{w}_{n}^H {\textbf{h}}_{n'}|^2\right).$ A convex objective function characterizes the optimization problem described earlier, but its efficient solution is challenging due to the unitary constraints listed in Table \ref{tab:1}. One approach to this challenge is to relax the unitary constraint and solve the relaxed optimization problem. However, this results in a performance loss as the resulting solution is not based on the original problem. Therefore, careful consideration is required to balance the trade-off between performance and solution efficiency.
Additionally, insights derived from Table \ref{tab:1} underscore that single-connected and fully connected serve as exceptional cases within the group-connected architecture. Consequently, our initial focus is on developing an efficient solution for the group-connected architecture, intending to extend it to address other architectural configurations later. Specifically, $L=1$ represents the fully connected architecture, whereas $L=K$ represents the singly connected architecture.

For the sake of computational simplicity, we define $\hat{\zeta_n}=\sqrt{(1+\psi_n)p_n}\zeta_n\chi_n\Gamma_n$, and we rewrite $R_n\left(\boldsymbol{\Phi}\right)$ as:
\begin{equation}
\small
\label{Transform}
\begin{aligned}
&R_n\left(\boldsymbol{\Phi}\right)=
2\mathcal{R}\{\hat{\zeta_n}^{\dagger}\textbf{w}_{n}^H (\textbf{h}_{n,b}+\textbf{G}\Phi \textbf{h}_{n,u})\}\\
&-|\zeta_n|^2\left({\sum}_{{n'}\ne n}^Np_{n'}|\textbf{w}_{n}^H \textbf{h}_{{n'},b}+\textbf{G}\Phi \textbf{h}_{{n'},u}|^2\right).
\end{aligned}
\end{equation}
Furthermore, the equation suggests that the IRS redirects signals uniformly from users to an AP, implying an equivalence between maximizing the minimum achievable rate (max-min) and maximizing the sum of rates (max-sum). Introducing the notation:
\begin{equation}
\small
\begin{aligned}
\!\!\mathcal{A}\! =\! \sum_{n=1}^N\!\! \left[\!2\mathcal{Re}\{\hat{\zeta}_n^{\dagger}\boldsymbol{w}_{n}^H\mathbf{h}_{n,b}\}\! -\! |\zeta_n|^2 \sum_{{n'}\neq n}^N\! p_{n'}\! |\boldsymbol{w}_{n}^H \mathbf{h}_{{n'},b}|^2 \right].
\end{aligned}
\end{equation}

We also define the transformed vector $\hat{\boldsymbol{w}}_n = \boldsymbol{w}_{n}^H\mathbf{G}$, and introduce matrices $\mathbf{X}$, $\mathbf{Y}$, and $\mathbf{Z}$ within the same framework: $\mathbf{X} = \sum_{n=1}^N \hat{\zeta}_n^{\dagger}\mathbf{h}_{n,u}\hat{\boldsymbol{w}}_n$, $\mathbf{Y} = \sum_{n=1}^N \chi_n \Gamma_n |\zeta_n|^2 \hat{\boldsymbol{w}}_n^H \hat{\boldsymbol{w}}_n$, and $\mathbf{Z} = \sum_{{n'}\neq n}^N p_{n'} \mathbf{h}_{{n'},u} \mathbf{h}_{{n'},u}^H$, effectively consolidating the mathematical representation into the discourse. Substituting these values into Equation \eqref{Transform} redefines the optimization problem.
\begin{subequations}
\label{P2A43B}
\begin{align}
&\mathcal{P}_{2A4-3A}:  \max_{\boldsymbol{\Phi}}\mathcal{A}+2\mathcal{R}\{\text{Tr}(\Phi\boldsymbol{X})\}-\text{Tr}(\Phi\boldsymbol{Z}\Phi^H\boldsymbol{Y}),\\
&\boldsymbol{\Phi}\quad\text{as per Table \ref{tab:1}} .\notag
\end{align}
\end{subequations}

Therefore, to solve the optimization problem effectively, we break down the objective function into distinct groups. This allows us to restructure the objective function of \eqref{P2A43B} into a more manageable form, as demonstrated in the equation \eqref{OFRF}.

\begin{equation}
\small
\begin{aligned}
\label{OFRF}
&F=\mathcal{A}+{\sum}_{l=1}^L2\mathcal{R}\{\text{Tr}(\Phi_l\boldsymbol{X}_l)\}\\
&-\text{Tr}({\sum}_{{l'}=1}^L\Phi_{l'}{\sum}_{l=1}^L\boldsymbol{Z}_{{l'},l}\Phi{l}^H\boldsymbol{Y}_{{l},{l'}}).
\end{aligned}
\end{equation}
Similarly, for any group $l$ it can be expressed as: 
\begin{equation}
\small
\begin{aligned}
\label{OFRF1}
\mathcal{F}_l&=\mathcal{A}+2\mathcal{R}\{\text{Tr}(\Phi_l\boldsymbol{X}_l)\}-\text{Tr}({\sum}_{{l'}=1}^L\Phi_{l'}\boldsymbol{Z}_{{l'},l}\Phi{l}^H\boldsymbol{Y}_{{l},{l'}}),\\
&=\mathcal{A}+2\mathcal{R}\{\text{Tr}(\Phi_l\boldsymbol{X}_l)\}-\text{Tr}(\Phi_{l}\boldsymbol{Z}_{{l},l}\Phi{l}^H\boldsymbol{Y}_{{l},{l}})\\
&-\text{Tr}({\sum}_{{l'}=1}^L\Phi_{l'}\boldsymbol{Z}_{{l'},l}\Phi{l}^H\boldsymbol{Y}_{{l},{l'}}),\\
&=\mathcal{A}+\text{Tr}(\Phi_{l}\boldsymbol{Z}_{{l},l}\Phi{l}^H\boldsymbol{Y}_{{l},{l}})-2\mathcal{R}\{\text{Tr}(\Phi_l\boldsymbol{X}_l)\}\\
&+\text{Tr}({\sum}_{{l'} \ne l}^L\Phi_{l}\boldsymbol{Z}_{{l},{l'}}\Phi_{l'}^H\boldsymbol{Y}_{{l'},{l}}).\\
&\mathcal{F}_l=\mathcal{A}+\text{Tr}(\Phi_{l}\boldsymbol{Z}_{{l},l}\Phi{l}^H\boldsymbol{Y}_{{l},{l}})-2\mathcal{R}\{\text{Tr}(\Phi_l\hat{\boldsymbol{X}_l})\}.
\end{aligned}
\end{equation}
Where $\hat{\boldsymbol{X}_l}=\boldsymbol{X}_l-(\Phi_{l}\boldsymbol{Z}_{{l},{l'}}\Phi_{l'}^H\boldsymbol{Y}_{{l'},{l}})$. Based on this, we reformulate \eqref{P2A43B} for each group as follows as follows:
\begin{subequations}
\label{P2A43C}
\begin{align}
\small
\!\!\!\mathcal{P}_{2A4-3B}\!:& \min_{\boldsymbol{\Phi}_l} \mathcal{F}_l, \quad,\boldsymbol{\Phi}_l^H\boldsymbol{\Phi}_l=\textbf{I}_{\Bar{K}}, \forall l \in \mathcal{L} .
\end{align}
\end{subequations}
The \eqref{P2A43C} requires precise iterative calculations to find the optimal solution. To accomplish this, the gradient of equation \eqref{OFRF1} is first determined as $\boldsymbol{\Upsilon}=\nabla \mathcal{F}(\boldsymbol{\Phi}_l)=2\boldsymbol{Y}_{{l},{l}}\Phi_{l}\boldsymbol{Z}_{{l},l}-2\hat{\boldsymbol{X}_l}^H$. This gradient indicates the steepest descent direction and guides parameter adjustment. However, given the Riemannian geometry of \eqref{P2A43C}, it is also imperative to compute the Riemannian gradient at the point $\boldsymbol{\Phi}_l$. Subsequently, the Riemannian gradient is translated to the identity according to
\begin{equation}
    \label{RI}    \boldsymbol{\Xi}\left(\boldsymbol{\Phi}_l\right)=\boldsymbol{\Upsilon}_{\boldsymbol{\Phi}_l}{\boldsymbol{\Phi}_l}^H-{\boldsymbol{\Phi}_l}\boldsymbol{\Upsilon}_{\boldsymbol{\Phi}_l}^H.
\end{equation}
After calculating the Riemannian gradient, the value of of $\boldsymbol{\Phi}_l$ is updated using the exponential map as given below:
\begin{equation}
    \label{PU}
    \boldsymbol{\Phi}_l^{i+1}=\underbrace{\text{expm}\left(-\kappa\boldsymbol{\Xi}^{i}\right)}_{\text{Rotational matrix}}{\boldsymbol{\Phi}_l^i}.
\end{equation}
Where $\boldsymbol{\Xi}^{i}=\boldsymbol{\Xi}\left(\boldsymbol{\Phi}_l^{i}\right)$ and $\kappa\ge0$ serve as the convergence control parameter in equation \eqref{RI}. Additionally, it's essential to note that the matrix $\boldsymbol{\Xi}$ is skew-Hermitian, denoted as $\boldsymbol{\Xi}=-\boldsymbol{\Xi}^H$, ensuring that the updated value of $\boldsymbol{\Phi}_l^{i+1}$ complies with the unitary constraint as specified in equation \eqref{P2A43B}. However, the computational complexity arises from the $\text{expm}(\cdot)$ operation involved in calculating the rotational matrix. Hence, we adopt the Taylor approximation method to alleviate this complexity. For simplicity, we define $\boldsymbol{\varphi}^i$ as a rotation matrix, where $\boldsymbol{\varphi}^i=\text{expm}\left(-\kappa\boldsymbol{\Xi}^{i}\right)$. Specifically, the rotation matrix can be approximated as $\Bar{\boldsymbol{\varphi}}^i=\textbf{I}+\kappa\boldsymbol{\Xi}^i+{\kappa^2}/{2}(\boldsymbol{\Xi}^i)^2$. With this approximation, we can rewrite equation \eqref{PU} as:
\begin{equation}
\small
    \label{PU2}
    \begin{aligned}
    \boldsymbol{\Phi}_l^{i+1}&={\boldsymbol{\Phi}_l^i}-\kappa[\boldsymbol{\Upsilon}_{\boldsymbol{\Phi}_l^i}{\boldsymbol{{\Phi}_l^i}^H}{\boldsymbol{\Phi}_l^i}-{\boldsymbol{\Phi}_l^i}{\boldsymbol{\Upsilon}_{{\boldsymbol{\Phi}_l^i}}^H}{\boldsymbol{\Phi}_l^i}]\\
    &+\frac{\kappa^2}{2}\Bigl[\left(\boldsymbol{\Upsilon}_{\boldsymbol{\Phi}_l^i}{\boldsymbol{\Phi}_l^i}^H\right)^2\boldsymbol{\Phi}_l^i+\left(\boldsymbol{\Phi}_l^i{\boldsymbol{\Upsilon}_{\boldsymbol{\Phi}_l^i}^H}\right)^2\boldsymbol{\Phi}_l^i\\
    &-{\boldsymbol{\Upsilon}_{\boldsymbol{\Phi}_l^i}}{\boldsymbol{\Phi}_l^i}^H{\boldsymbol{\Phi}_l^i}{\boldsymbol{\Upsilon}_{\boldsymbol{\Phi}_l^i}^H}{\boldsymbol{\Phi}_l^i}-{\boldsymbol{\Phi}_l^i}{\boldsymbol{\Upsilon}_{\boldsymbol{\Phi}_l^i}^H}{\boldsymbol{\Upsilon}_{\boldsymbol{\Phi}_l^i}}{\boldsymbol{{\Phi}_l^i}^H}{\boldsymbol{\Phi}_l^i}\Bigr].
    \end{aligned}
\end{equation}
By employing this transformation, the multiplicative update in  \eqref{PU} is converted into an additive update. However, determining the appropriate step size $\kappa$ poses a challenge in practice due to the time-varying nature of ${\boldsymbol{\Phi}_l^i}$. An adaptive step size approach, such as the steepest descent algorithm combined with the Armijo rule, is necessary for selecting the optimal step size value. This ensures that the algorithm converges to a local optimum after several iterations. Using the third-order Taylor approximation method and adaptive step sizes, Algorithm \ref{optimizedAlgorithmLabel} explains the proposed approach.

\begin{algorithm2e}
\small
\SetAlgoLined
\SetKwInOut{Initialization}{Initialization}{\textbf{Initialization}: ${i,j}=1$, $\boldsymbol{\Phi}_l^{i}=\mathbf{I}$, and $\kappa=1$, $\boldsymbol{p}^o$,$\boldsymbol{w}^o$,$\boldsymbol{\Phi}^o$}\\
\KwOut{Optimal Values $\left(\{\boldsymbol{p}^*\},\{\boldsymbol{w}^*\},\{\boldsymbol{\Phi}^*\}\}\right)$}
\While{Error $\ge \epsilon$}{
Under the provided task segmentation and computational resource values from Algorithm \ref{algo:Algorithem 1}, solve \eqref{APB} and \eqref{p2A43A} to obtain  $\boldsymbol{w}^j$, and $\boldsymbol{p}^j$.\\
\For{each group l $\in \mathcal{L}$}{
\While{i$\le i^{max}$}{
${\boldsymbol{\Upsilon}{\boldsymbol{\Phi}_l^i}}=\frac{\partial \mathcal{F}}{\partial {\boldsymbol{\Phi}_l^*}}\left({\boldsymbol{\Phi}_l^i}\right)$;
$\boldsymbol{\Xi}^{i}={\boldsymbol{\Upsilon}{\boldsymbol{\Phi}_l^i}} {\boldsymbol{\Phi}_l^i}^H-{\boldsymbol{\Phi}l^i} {\boldsymbol{\Upsilon}{\boldsymbol{\Phi}l^i}^H}$\;
\If{$\left|\boldsymbol{\Xi}^{i}\right|{\boldsymbol{\Phi}_l^i}^2$ is sufficiently small}{
\textbf{STOP};
}
$\Bar{\boldsymbol{\varphi}}^i=\mathbf{I}-\kappa \boldsymbol{\Xi}^{i}+\left(\kappa \boldsymbol{\Xi}^{i}\right)^2 / 2-\left(\kappa \boldsymbol{\Xi}^{i}\right)^3 / 6$;
$\tilde{\mathbf{Q}}_i=\Bar{\boldsymbol{\varphi}}^i \Bar{\boldsymbol{\varphi}}^i$;\\
\While{$\mathcal{F}\left({\boldsymbol{\Phi}_l^i}\right)-\mathcal{F}\left(\tilde{\mathbf{Q}}_i {\boldsymbol{\Phi}l^i}\right) \geq \kappa\left|\boldsymbol{\Xi}^{i}\right|{\boldsymbol{\Phi}_l^i}^2$}{
$\Bar{\boldsymbol{\varphi}}^i:=\tilde{\mathbf{Q}}_i$;
$\tilde{\mathbf{Q}}_i=\Bar{\boldsymbol{\varphi}}^i \Bar{\boldsymbol{\varphi}}^i$;
$\kappa:=2 \kappa$;
}
\While{$\mathcal{F}\left({\boldsymbol{\Phi}_l^i}\right)-\mathcal{F}\left(\Bar{\boldsymbol{\varphi}}^i {\boldsymbol{\Phi}l^i}\right) \leq (\kappa/2)\left|\boldsymbol{\Xi}^{i}\right|{\boldsymbol{\Phi}_l^i}^2$}{
$\Bar{\boldsymbol{\varphi}}^i=\mathbf{I}-\kappa \boldsymbol{\Xi}^{i}+\left(\kappa \boldsymbol{\Xi}^{i}\right)^2 / 2-\left(\kappa \boldsymbol{\Xi}^{i}\right)^3 / 6$,
$\kappa:=\kappa / 2$;
}
${\boldsymbol{\Phi}_l^{i+1}}=\Bar{\boldsymbol{\varphi}}^i {\boldsymbol{\Phi}_l^i}$;
$i:=i+1$;
}
\textbf{Output:}${\boldsymbol{\Phi}_l^*}$
}
$\mathbb{F}^j\leftarrow$solve \eqref{OF2A3}.\\
Error = |$\mathbb{F}^j-\mathbb{F}^o|$.\\
\textbf{Update:} Auxiliary Variables using \eqref{AUXX1}, \eqref{AUX1} and \eqref{AUX2}. 
}
\caption{\textbf{Iterative Algorithm to solve} \eqref{P2A3}}
\label{optimizedAlgorithmLabel}
\end{algorithm2e}

\subsubsection{Step 2} Lastly, the $\chi_n$ and $\Gamma_n$ is updated.  As Per Lemma \ref{lemma-3}, the optimal values of $\boldsymbol{p}^*,\boldsymbol{w}^*,\boldsymbol{\Phi}^*$ satisfy:
\begin{equation}
\begin{aligned}
{R_n}\left({p}_n^*,\boldsymbol{w}_n^*,\boldsymbol{\Phi}^*\right)\chi_n-1=0    , \forall n \in \mathcal{N}.\\
{R_n}\left({p}_n^*,\boldsymbol{w}_n^*,\boldsymbol{\Phi}^*\right)\Gamma_n-{\beta_nS_n}=0    , \forall n \in \mathcal{N}.
\end{aligned}
\end{equation}
Then, for the next iteration,  auxiliary value can be updated as:
\begin{equation}
\small
\label{AUXX1}
\begin{aligned}
\chi_n^{j+1}=\left(1+\Theta^j\right)\chi_n^{j}+\frac{\Theta^j}{R_n\left({p}_n^*,\boldsymbol{w}_n^*,\boldsymbol{\Phi}^*\right)}    , \forall n \in \mathcal{N}.\\
\Gamma_n^{j+1}=\left(1+\Theta^j\right)\Gamma_n^{j}+\Theta^j\frac{\beta_nS_n}{R_n\left({p}_n^*,\boldsymbol{w}_n^*,\boldsymbol{\Phi}^*\right)}    , \forall n \in \mathcal{N}.
\end{aligned}
\end{equation}
The Procedure is summarised in \textbf{Algorithm} \ref{optimizedAlgorithmLabel}.
\subsubsection{Initialization}
Algorithm \ref{optimizedAlgorithmLabel} operates iteratively to find optimal solutions, and its performance and convergence depend significantly on initial values. The initialization begins by placing the UAV at the region's center and equally allocating transmission power. Subsequently, heuristics are applied to adjust the power to satisfy constraints \eqref{P0C02} and \eqref{P0C03}. Initializing the IRS phase shift as an identity matrix. After this initialization, given the values of communication resources, calculations are performed as outlined in \eqref{APB}, \eqref{AUX1}, and \eqref{AUX1}. The resulting solution is the initial guess for the subsequent steps of Algorithm \ref{optimizedAlgorithmLabel}.
\subsubsection{Convergence and  Complexity Analysis of Algorithm \ref{optimizedAlgorithmLabel}}
\paragraph{Complexity Analysis}
\label{Comp2}
This section provides per-iteration complexity analysis of Algorithm \ref{optimizedAlgorithmLabel}, which iteratively addresses  \eqref{APB}, \eqref{p2A43A}, and \eqref{P2A43C}, to find the optimal solution. To begin,  the computational complexity for  \eqref{APB}  depends on both  ($N$) and ($M$), leading to a complexity of $\mathcal{O}(NM^3)$. Moreover, \eqref{p2A43A} presents a convex optimization problem involving $N$ decision variables and $2N$ constraints resulting in the computational complexity $\mathcal{O}(2N^4)$. Lastly, the complexity related to solving \eqref{P2A43C} can be expressed as $\mathcal{O}(LI^{max}\Bar{K}^3)$, with $I^{max}$ representing the maximum number of iterations required by the gradient algorithm. To provide a comprehensive perspective, of total $i^{max}$ iterations, the overall computational complexity of \textbf{Algorithm} \ref{optimizedAlgorithmLabel} expressed as $\mathcal{O}\left(i^{max}\left((N) + (NM^3) + (2N^4) + (LI^{max}\Bar{K}^3)\right)\right)$.
\paragraph{Convergence Analysis}
Let $\mathcal{Q}=[\boldsymbol{p},\boldsymbol{w}^*,\boldsymbol{\Phi}^*,\boldsymbol{\psi}^*,\boldsymbol{\zeta}^*]$ represent best achievable values for \eqref{p3b} at the current iteration. The objective function value of \eqref{p3b} at this iteration is denoted as $\mathbb{F}(\mathcal{Q}_j)$, which corresponds to a subset of the original optimization problem, which iteratively improves as per observation made by Beck et al. \cite{beck2010sequential}, i.e., $\mathbb{F}(\mathcal{Q}_{i+1}) \geq \mathbb{F}(\mathcal{Q}_{i})$. Convergence occurs when the difference between iterations falls below a threshold $\epsilon$, i.e., $\mathbb{F}_{i+1} - \mathbb{F}_{i} \leq \epsilon$.
\begin{table}[t]
    \centering
    \caption{Simulation Parameters}
    \label{tab:SP}
    \begin{tabular}{@{}l@{}rl@{}@{}@{}r@{}}
    \hline \hline
    Parameters & Values & Parameters & Values \\ \hline \hline
    N & [2 - 20] Users & $C^{'}_n$ & [0.5 - 100] Kcycles/bit \\
    K & 32 - 64 Elements & M & 64 \\
    B & 20MHz & $N_o$ & -174 dBm/Hz \\
    $PL_o$ & 31.5 dB & $S_n$ & [100-200] Kbits \\
    $\{f_n^{max},f_s^{max}\}$ & \{0.1, 20\}$\times 10^10$ cycles/s & $E^{Max}$ & 2 Joules \\
    \hline
    \end{tabular}
\end{table}

\section{Results and Discussion}
\label{RD}
\begin{figure}[!t]
    \centering
    \begin{subfigure}[t]{0.5\textwidth}
        \centering
        \includegraphics[width=0.8\linewidth]{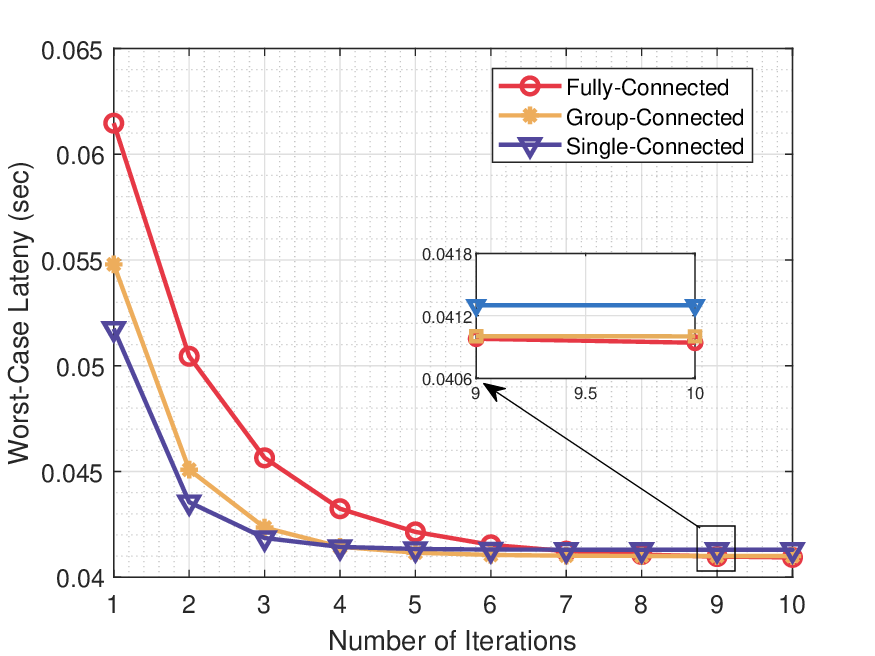}
        \vspace{-1mm}
        \caption{Algorithm 1}
        \label{A1}
    \end{subfigure}%
    \vspace{-0.5mm}
    \\
    \begin{subfigure}[t]{0.5\textwidth}
        \centering
        \includegraphics[width=0.8\linewidth]{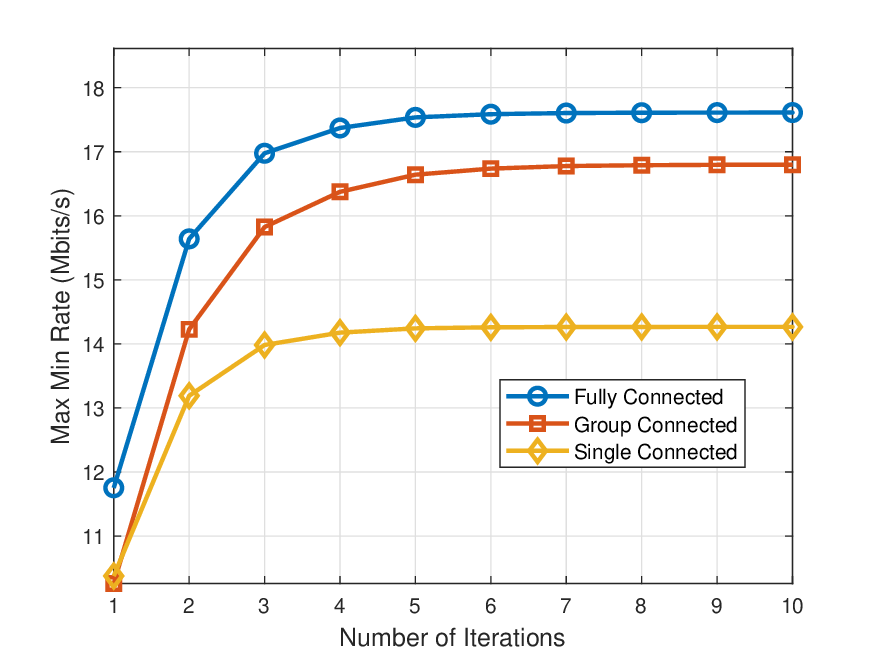}
        \vspace{-1mm}
        \caption{Algorithm 2}
         \label{A2}
    \end{subfigure}
    \vspace{-1mm}
    \caption{Convergence Analysis}
    \label{CA}
\end{figure}
This section presents numerical results to demonstrate the effectiveness of the {\color{black}BD-IRS-UAV} enabled MEC network. To showcase the efficacy of the proposed scheme, we conducted extensive Monte-Carlo simulations and calculated average results using the simulation parameters listed in Table \ref{tab:SP}. In our simulations, we randomly distributed $N$ users across a $150m\times150m$ area, with $[d_{n,u}^{max},d_{u,b}^{max}]=100m$ and $H=80m$.  As for the communications channel, we consider both the small-scale fading and the large-scale path loss. More explicitly, the small-scale fading is assumed to be independent and identically distributed (i.i.d.) and obey the complex Gaussian distribution associated with zero mean and unit variance, while the path loss in dB is given by
$
    PL=PL_o+10\alpha\log_{10}\left(\frac{d}{d_o}\right)+\eta,  
$
where $PL_o$ represents the path loss at reference distance $d_o=1$, $\alpha$ and $d$ represent the path-loss exponent and distance of the communication link. Moreover, $\eta$ represents the additional path loss components that depend on the probability of line-of-sight communication. For example, in an urban environment, $\eta=\{1,20\}$, $1$ for line-of-sight communication and $20$ otherwise \cite{mahmood2023joint}. Moreover, in our case, we consider that users communicate with the AP through the direct line with a $5\%$ chance of a line-of-sight communication link. In contrast, communication between the user and the AP through UAV-IRS is carried out with $95\%$ probability of line-of-sight communication.
In addition, we evaluate the performance of the proposed scheme based on several Key Performance Indicators (KPIs), including achievable rate and latency, and compare these outcomes to their respective benchmark schemes.
\par
In the first section of the proposed framework, we perform the optimal allocation of computational resources by iteratively solving equations \eqref{P1A} and \eqref{tasksegmentation}, given the available communication resources, to minimize maximum latency, as detailed in Section \ref{OCS}. Additionally, we evaluate the performance of the proposed scheme by comparing it with the following reference approaches:
\begin{itemize}
    \item \textit{Binary Offloading}: This scheme enforces the task to be computed locally or at the MEC server, i.e., $\beta \in \{0,1\}$. 
    \item \textit{Edge Computation}: Tasks are entirely computed at the MEC server in this scheme, i.e., $\beta =0$. 
    \item \textit{Fixed Computation}: In this scheme, only edge computational resources are optimized by allocating the entire local resources \cite{9651523}. 
   \item \textit{Local}: In this scheme, optimal allocation of  local computational resources is carried out, i.e., $\beta =1$
\end{itemize}
Once the computational resources are established, we optimize the allocation of communication resources by iteratively solving equations \eqref{APB}, \eqref{p2A43}, and \eqref{P2A43C} to maximize the minimum rate while satisfying relevant constraints. Additionally, we conduct a comparative analysis to demonstrate the efficacy of various IRS architectures in reflective mode, evaluating the influence of optimal power allocation, phase shift control, and optimal placement of BD-IRS-UAV.
\begin{itemize}
    \item \textit{Fully Connected}: In this architecture, all elements of the IRS are interconnected, which leads to higher gain. In our proposed framework, achieving a fully connected architecture involves solving \eqref{P2A43C} while setting the number of groups equal to 1, i.e., $L=1$.
    \item \textit{Group Connected}: In this scheme, the total IRS elements are divided into $L$ distant groups, with each group forming a fully connected architecture and solved using \eqref{P2A43C}. 
    \item  \textit{Single Connected}: In contrast, in a single-connected architecture, the elements are not connected to each other. In our proposed framework, achieving a single connected architecture is accomplished by setting the number of groups equal to the number of elements, i.e., $L=K$.
\end{itemize}
In addition, we conduct Monte-Carlo simulations to obtain averaged results based on independent channel realizations.
\subsection{Performance Analysis}
Fig. \ref{CA} provides a comprehensive view of the convergence analysis of the proposed algorithms under varying simulation parameters. Specifically, Fig. \ref{A1} presents the convergence analysis of Algorithm \ref{algo:Algorithem 1}, using worst-case latency as the performance metric. These results demonstrate a remarkable convergence to the stable point as a result of iterative updates to the task segmentation variable, which facilitate the optimal allocation of computational resources. Furthermore, Fig. \ref{A2}  shows the convergence behavior of algorithm \ref{optimizedAlgorithmLabel}, which iteratively solves equations  \eqref{APB}, \eqref{p2A43}, and \eqref{P2A43C}. After some iterations, the results affirm its convergence towards a stable and feasible solution. It is essential to emphasize that the convergence time varies significantly among different IRS architectures due to their unique counts of non-zero IRS elements. Consequently, each architecture necessitates different times to reach viable solutions that effectively minimize network latency and maximize the data rate.
\subsection{Data Rate and Latency Analysis}

\begin{figure}[!t]

    \centering
    \begin{subfigure}[t]{0.5\textwidth}
        	\centering
	\includegraphics[width=0.8\linewidth]{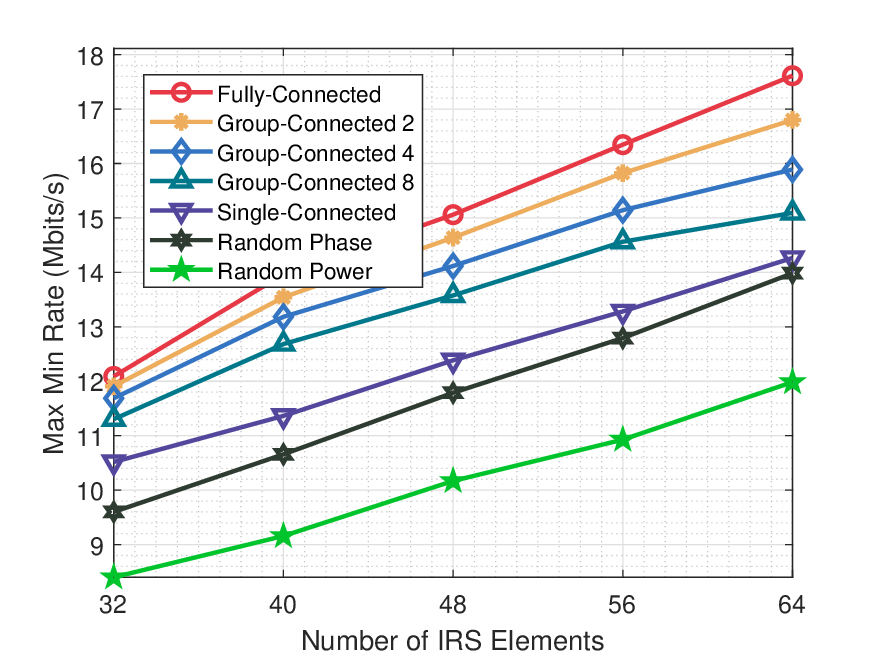}
 \vspace{-2mm}
	\caption{Optimal Placement over UAV}
	\label{fig:R2A}
    \end{subfigure}%
    \vspace{-0.5mm}
    \\
    \begin{subfigure}[t]{0.5\textwidth}
        	\centering
	\includegraphics[width=0.8\linewidth]{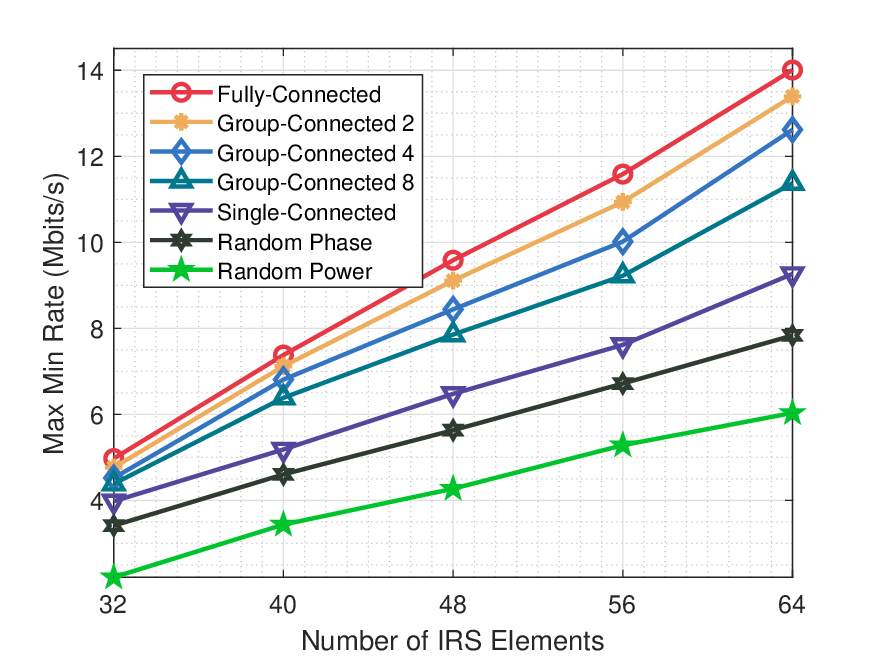}
 \vspace{-2mm}
	\caption{Fixed Placement over Building}
	\label{fig:R2B}
    \end{subfigure}
    \vspace{-1mm}
	\caption{Impact of Phase Shift on System Performance}
	\label{fig:R2}
\end{figure}

\begin{figure*}[!t]
    \centering
    \begin{subfigure}[t]{0.32\textwidth}
        \centering
        \includegraphics[height=1.7in]{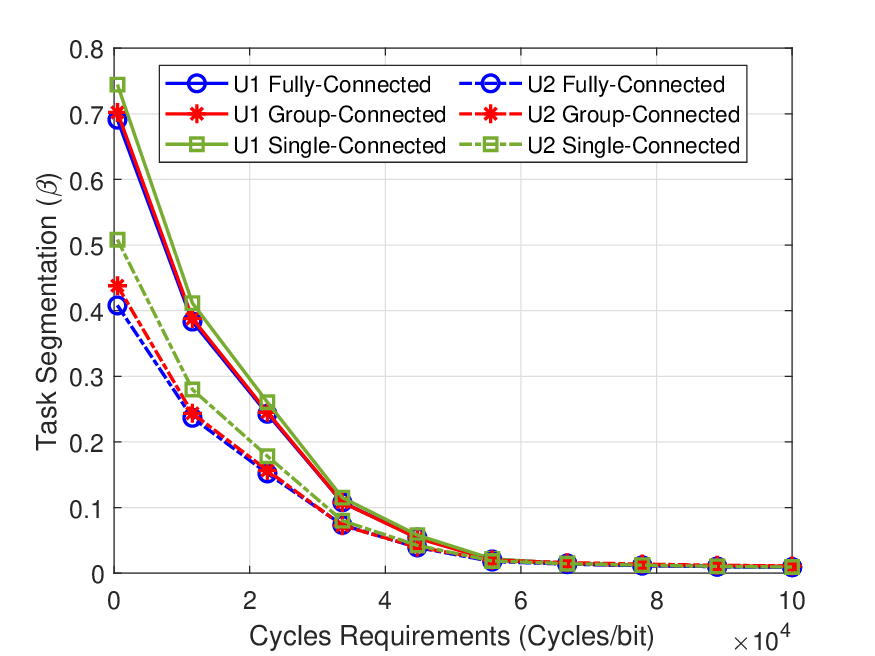}
        \caption{Task segmentation}
        \label{R4a}
    \end{subfigure}%
    ~ 
    \begin{subfigure}[t]{0.32\textwidth}
        \centering
        \includegraphics[height=1.7in]{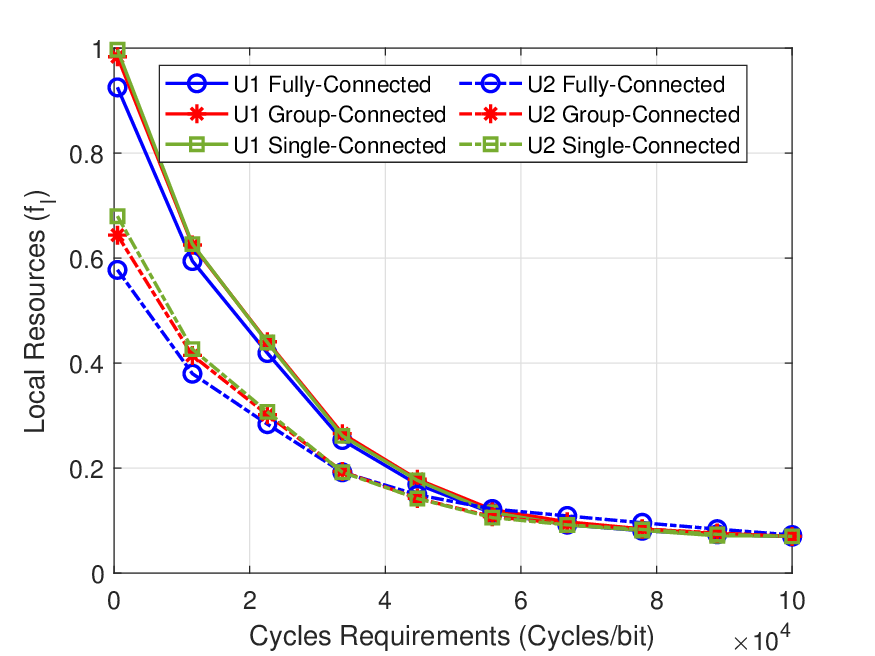}
        \caption{Local Resources}
        \label{R4b}
    \end{subfigure}
      ~ 
    \begin{subfigure}[t]{0.32\textwidth}
        \centering
        \includegraphics[height=1.7in]{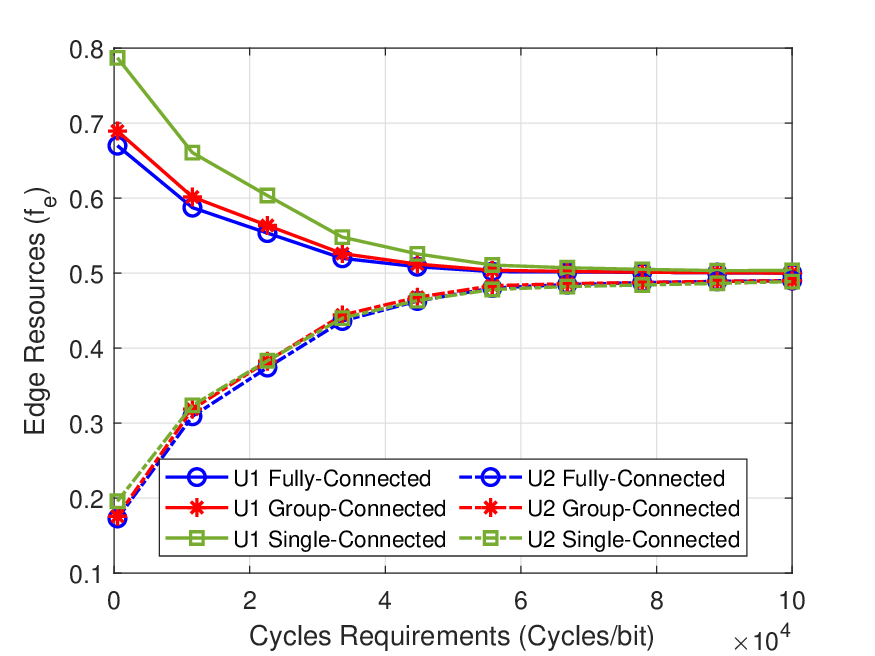}
        \caption{Edge Resources}
        \label{R4c}
    \end{subfigure}
    \vspace{-1mm}
    \caption{Task Segmentation and Resource Allocation Overview,  U1: User 1, U2: User 2}
    \label{R4}
\end{figure*}
\begin{figure}[h]
	\centering
	\includegraphics[width=0.8\linewidth]{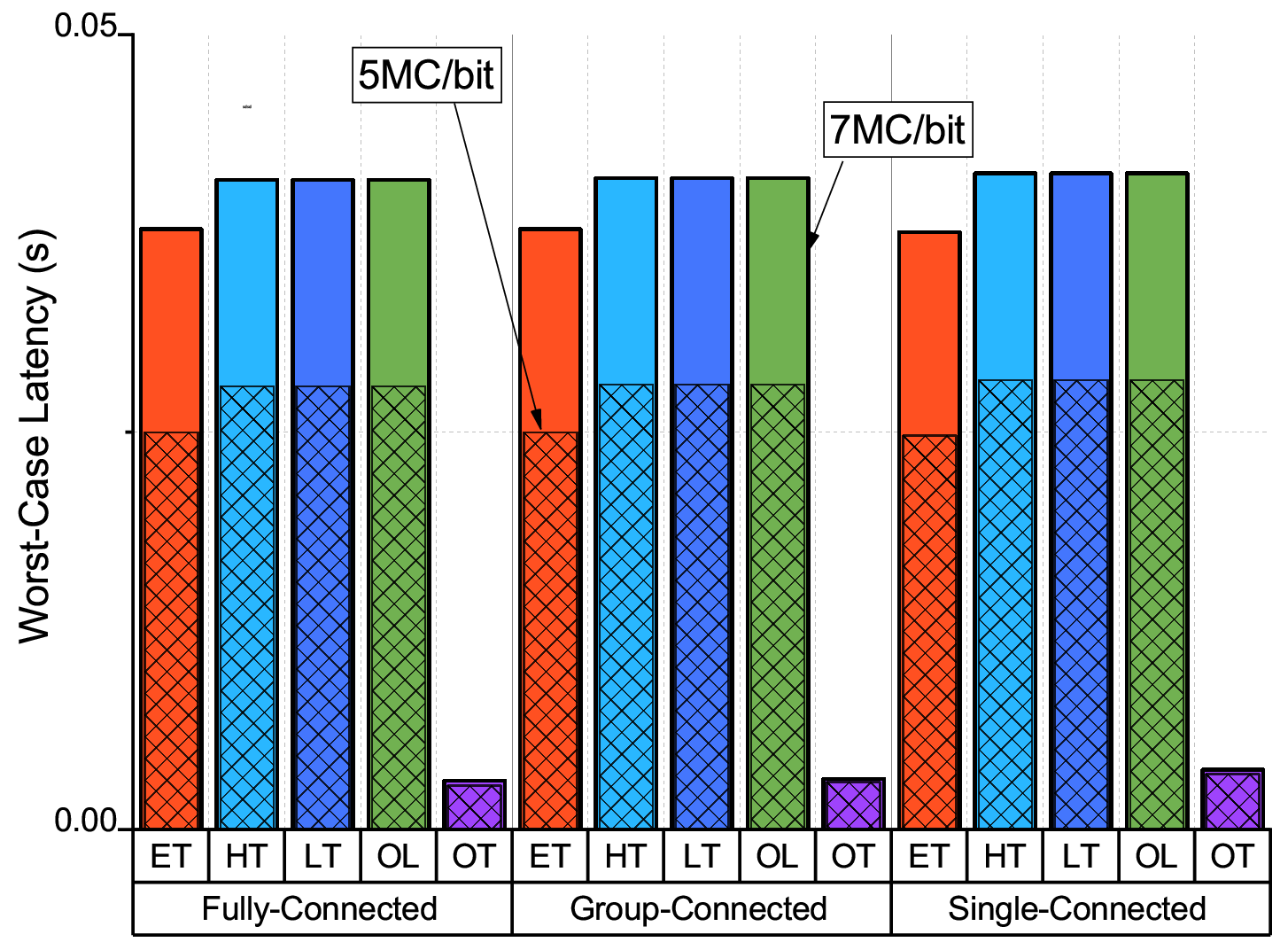}
 \vspace{-1mm}
	\caption{System Performance, Where ET: Edge Time, HT: Hovering Time, LT: Local time, OL: Overall Latency, OT: Offloading Time}
	\label{fig:R5}
\end{figure}
\par
Fig. \ref{fig:R2} presents a comparative analysis of the proposed schemes across various IRS architectures, focusing on the data rate metric. The results demonstrate a clear trend: as the number of IRS elements increases, the minimum data rate experienced by users also increases. This trend is particularly notable when comparing the performance of random phase shift and random power approaches in both UAV-enabled IRS and fixed deployment over the building. Furthermore, the optimal deployment of IRS in both fully connected and group-connected modes yields performance improvements of approximately $25.76\%$ and $25.43\%$, respectively, compared to the conventional fixed deployment of IRS on buildings. These findings emphasize the significance of optimal phase shift, transmission power control, and optimal IRS placement in next-generation wireless communication networks, particularly in dense urban environments where high-rise buildings frequently obstruct direct line-of-sight communication.
\par
In addition, we conducted a comparative analysis across various IRS architectures, revealing that fully connected and group-connected architectures consistently outperform the traditional single-connected architecture by approximately  $23.45\%$ and $17.75\%$, respectively. The superiority of fully connected architectures stems from their significantly higher number of non-zero elements, resulting in higher channel gain. It is noteworthy that fully connected architectures, while delivering exceptional performance, also entail higher complexity, as discussed in Section \ref{Comp2}. It is also shown that as the number of elements per group in a group-connected architecture increases, its performance approaches that of fully connected architectures, resulting in lower complexity. This observation highlights a fundamental tradeoff between performance and complexity that must be carefully managed during system design.

\begin{figure}[h]
	\centering
	\includegraphics[width=0.8\linewidth]{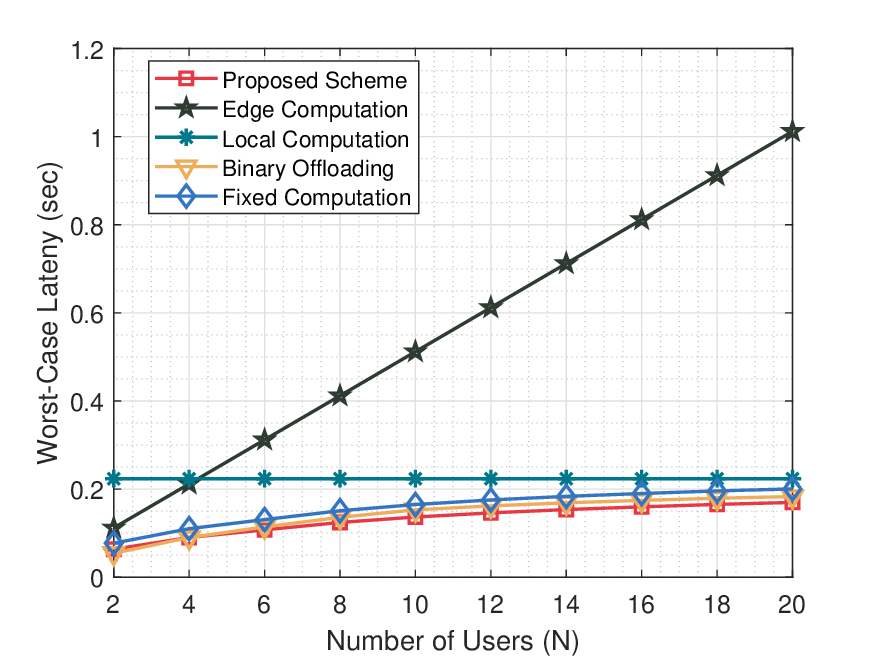}
 \vspace{-1mm}
	\caption{Latency Analysis }
	\label{fig:R6}
\end{figure}

\begin{figure}[h]
	\centering
	\includegraphics[width=0.8\linewidth]{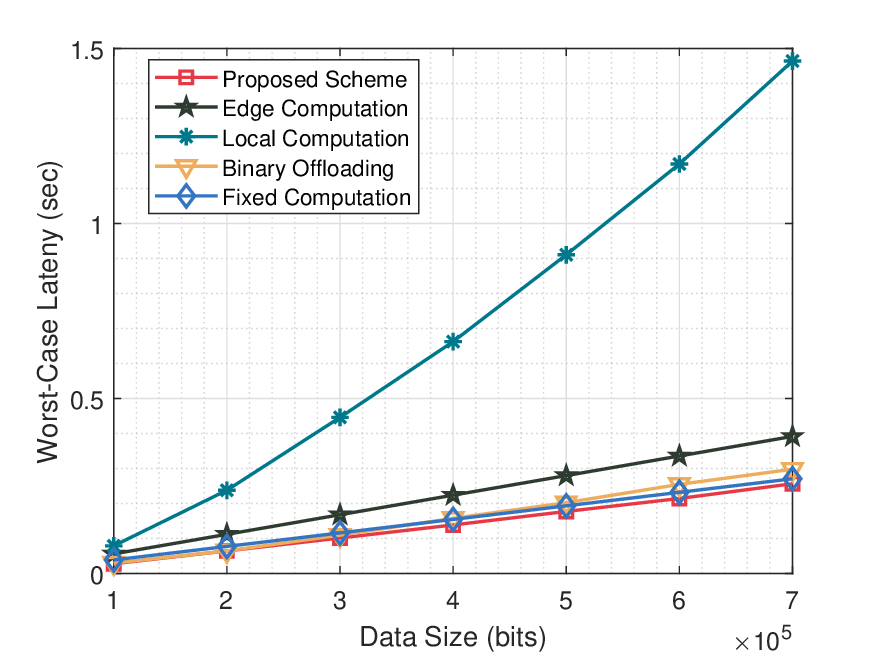}
 \vspace{-1mm}
	\caption{UAV Hovering Time}
	\label{fig:R7}
\end{figure}
\par
Next, we analyzed the effect of different IRS architectures on the MEC network's performance across different computational cycle requirements, as illustrated in Fig. \ref{R4}. Our simulations focused on a scenario with two users, U1 and U2, where U1 had more sensed bits to process than U2, specifically $S_1 > S_2$. Fig. \ref{R4a} illustrates that, with a small number of cycle requirements, tasks are effectively partitioned. Conversely, $\beta_n$ approaches zero as the computational cycle requirements increase. This leads to a decrease in local computational resources, while edge computational resources begin to increase, as demonstrated in Fig. \ref{R4c} and Fig. \ref{R4b}, respectively. This trend arises because computing bits with small cycles requires less energy and vice versa, as per \eqref{local Energy}. Therefore, to adhere to the constraint in \ref{P0C03}, devices begin offloading data to the MEC server for extensive computational processing.

\par
The results in Fig. \ref{R4a} underscore a preference for local computation of tasks with a maximum number of bits in our scenario. In this context, $S_1 > S_2$ leads to $\beta_1 > \beta_2$. This preference arises because offloading a smaller number of bits to the MEC server consumes less energy. Conversely, users with larger bit requirements experience increased latency when offloading compared to others. As a result, to ensure fairness, MEC allocates more resources to UE, having more computing data, as illustrated in Fig. \ref{R4b} and Fig. \ref{R4c}.
Additionally, it has been observed from \eqref{OF1} that the offloading time of tasks can be reduced by increasing the offloading data rate for specific users using advanced IRS architectures, such as fully connected or group-connected architectures. Results demonstrate that employing advanced IRS architectures results in higher data offloading to the MEC server than traditional single-connected architectures. This, in turn, leads to the allocation of more computational resources to other users, minimizing latency and maintaining fairness.

\par
Fig. \ref{fig:R5} provides detailed impacts of key system parameters. The results reveal that both fully connected and group-connected architectures result in minimal offloading times, leading to more efficient allocation of computational resources and, consequently, a reduction in overall system latency even for higher computational cycle requirements. Notably, the performance of fully connected and group-connected architectures exhibits similarities. However, it is worth noting that the fully connected architecture exhibits greater complexity due to its processing of a larger number of elements. This observation highlights the effectiveness of the group-connected architecture in terms of complexity compared to the fully connected architecture.
\par
Next, we delve into the performance analysis of our proposed scheme, which leverages optimal allocation of computation and communication resources through a partial offloading approach using the group-connected architecture of the IRS as shown in Fig. \ref{fig:R6}. The results demonstrate the superiority of our proposed approach by considering latency as a performance metric compared to others. When the number of users is small, the proposed scheme slightly outperforms the binary-offloading counterpart. And the performance gain increases notably as the number of users increases. This is because the user count rises and the per-user resource allocation decreases, leading to congestion at the MEC server.
In contrast, our proposed scheme efficiently partitions tasks into segments, effectively sidestepping congestion and minimizing system latency by approximately $13.44\%$ relative to the binary offloading scheme. These results underscore the effectiveness of our task segmentation approach compared to binary offloading, full edge offloading, and fixed computational resource allocation schemes. In addition, simulations were conducted to assess the performance across various data size (bits) requirements. The results, illustrated in Fig. \ref{fig:R7}, unequivocally highlight the efficacy of the proposed scheme.

\section{Conclusion}
\label{CC}
This paper has provided a new architecture combining UAVs and BD-IRS in the context of MEC networks and proposed an optimization framework to improve system performance. In particular, the BD-IRS-UAV placement, local and edge computational resource allocation, task segmentation, power allocation, received beamforming vector, and IRS 2.0 phase shift design were jointly optimized to minimize the worst-case system latency. The joint optimization problem was decoupled into two subproblems due to non-convexity/non-linearity and NP-hardness. Then, each subproblem was transformed into convex optimization and solved using standard convex methods. The proposed MEC framework was evaluated using the fully connected and group-connected architecture of BD-IRS and compared with the conventional single-connected diagonal IRS architecture. Numerical results demonstrated the superiority of the proposed BD-IRS-UAV enabled MEC network compared to conventional diagonal IRS in terms of worst-case latency, rate, and energy consumption.   
\begin{appendices}
\section{Proof of LEMMA \ref{lemma-beta}}
\label{APA}
Given the values of $\boldsymbol{f^s}$ and $\boldsymbol{f^l}$, the task segmentation value is determined by analyzing the monotonic relationship between the local computation time and the offloading time of the tasks, as shown in equations \eqref{local time}, \eqref{offloading time}, and \eqref{Edge time}. As the value of $\beta_n$ increases, the local computation time of the task decreases, leading to an increase in the offloading time. Based on this observation, we define the computational delay as $T_n(\hat{\beta}_n) = \max\{t_n^l, t_n^E\}$. The minimal delay occurs when $t_n^l = t_n^E$. By utilizing these formulas and conducting further calculations, we define the task segmentation value as follows:
\begin{equation}
\small
\beta_n=\frac{f_n^lf_n^{max}\left(R_nC_n+f_n^sf_s^{max}S_n\right)}{R_nC_n\left(f_n^lf_n^{max}+f_n^sf_s^{max}\right)+S_nf_n^lf_n^{max}f_n^sf_s^{max}}
\end{equation}
This end the proof. 

\end{appendices}
\bibliographystyle{IEEEtran}
\bibliography{ReferenceBibFile}
\end{document}